\documentclass[twoside,11pt,letter]{article}
\usepackage{times}
\usepackage{amsmath}
\usepackage{amsfonts}
\usepackage{amssymb}
\usepackage{amsthm}
\usepackage{graphicx}
\usepackage{mathrsfs}
\usepackage{color}
\usepackage{hyperref}
\usepackage{url}

\newtheorem{theorem}{Theorem}[section]
\newtheorem{proposition}[theorem]{Proposition}
\newtheorem{lemma}[theorem]{Lemma}
\newtheorem{corollary}[theorem]{Corollary}

\theoremstyle{definition}
\newtheorem{definition}[theorem]{Definition}
\newtheorem{example}[theorem]{Example}
\newtheorem{remark}[theorem]{Remark}

\newcommand{\bC}{\mathbb{C}}

\newcommand{\bS}{\mathbb{S}}
\newcommand{\cH}{\mathcal{H}}
\newcommand{\cD}{\mathcal{D}}
\newcommand{\cE}{\mathcal{E}}
\newcommand{\cP}{\mathcal{P}}
\newcommand{\cS}{\mathcal{S}}

\newcommand{\eP}{\mathsf{P}}
\newcommand{\ket}[1]{|#1\rangle}
\newcommand{\bra}[1]{\langle#1|}
\newcommand{\tr}{\mathrm{tr}}
\newcommand{\hist}{\boldsymbol{\Omega}}
\newcommand{\histb}{\boldsymbol{\Theta}}

\title{\textbf{Consistent histories and the ultrametric\\ on infinite tensor products of Hilbert spaces}}
\author{\textbf{Andrew Lesniewski}\\
Department of Mathematics\\
Baruch College\\
One Bernard Baruch Way\\
New York, NY 10010\\
USA}
\date{\today}

\begin{document}
\maketitle

\begin{abstract}
The consistent histories formulation of quantum mechanics assigns classical probabilities to sequences of coarse-grained events when the corresponding interference terms vanish. In finite models these terms generically need not vanish exactly, and their smallness alone supplies no asymptotic geometry on the space of histories. We construct such a geometry for histories that accumulate quantum records. When the normalized Dowker - Halliwell interference between two histories factorizes as the cumulative fidelity of their records, its decay is controlled by the accumulated squared Bures distance. In the infinite registration limit, the Kakutani - von Neumann dichotomy gives a sharp sector alternative: summable Bures defect yields weakly equivalent record sequences, whereas nonsummable defect produces orthogonal
superselection sectors and vanishing normalized interference. The convergence exponent of the Bures defect defines a pseudo-ultrametric on histories. After quotienting its zero distance degeneracy, it becomes a complete ultrametric, requiring neither compactness of the outcome space nor continuity of the record family, and coincides with the gauge invariant sector metric $\tilde d$ introduced in \cite{LesTP}. Its value, the \emph{consistency rate}, has three equivalent interpretations: it is the sector separation metric, the polynomial growth exponent of the negative logarithm of normalized interference, and, by the Helstrom formula, the corresponding exponent of the negative logarithm of the minimum error probability for discriminating the cumulative records. Classical separation profiles determine the rate: positive long time mean defect gives the maximal value $\delta=1$, while the borderline profile $d_m^2\asymp m^{-1}$ produces distinct superselection sectors at zero polynomial rate, $\delta=0$. Finally, we relate this construction to the isolated many-body branching mechanism of \cite{PCRR26}: their finite dimensional permutation sector records obey the same accumulated fidelity structure, while the infinite tensor product provides an asymptotic idealization of continued record accumulation in which sector separation becomes exact.
\end{abstract}

\section{Introduction}\label{sec:intro}

Two recent developments, arrived at independently and from opposite directions, meet at the same mathematical structure: the accumulation of quantum records along a history.

The physical backdrop is decoherence \cite{S19}. A quantum system becomes entangled with degrees of freedom that retain records of macroscopically different alternatives, and the growing distinguishability of those records suppresses interference between the corresponding branches. Conventionally, the record bearing degrees of freedom are identified with an external environment. A recent construction of Pilatowsky-Cameo, Cotler, Ranard, and Riedel \cite{PCRR26} shows explicitly how the same mechanism can emerge within an isolated many-body system. In their weakly disordered chaotic kicked top, the collective spin sector $\cS$ carries the quasiclassical variable, while the permutation sector $\cP$ acts as an internal bath. A branch has the schematic form
\begin{equation}\label{eq:PCRR-branch}
\ket{\phi_{\hist}}=\ket{\Omega_M}_{\cS} \otimes \ket{R_{\hist}}_{\cP},
\end{equation}
and the microscopic records obey, in the appropriate semiclassical regime, a finite bath fidelity law of the form
\begin{equation}\label{eq:PCRR}
|\langle R_{\histb}|R_{\hist}\rangle|^2\sim\exp\!\big(-\Gamma D^{(2)}_{\hist,\histb}\big),
\end{equation}
where
\begin{equation}\label{eq:PCRR-D2}
D^{(2)}_{\hist,\histb}\approx \frac12\sum_{m=1}^{M}\big(d(\Omega_{m-1},\Theta_{m-1})^2+d(\Omega_m,\Theta_m)^2\big).
\end{equation}
Thus an accumulated squared distance between two quasiclassical histories controls the loss of record fidelity and, through it, the suppression of interference in collective observables.

From the structural direction, in \cite{LesTP} we equipped von Neumann's weak equivalence classes of infinite tensor product vectors \cite{vN39} with a gauge invariant convergence-exponent pseudo-ultrametric $\tilde d$. After quotienting its zero distance degeneracy, the resulting ultrametric space is complete. For product vectors, the elementary defect is
\begin{equation*}
1-|\langle r_m,r'_m\rangle|=\frac12\,d_B(r_m,r'_m)^2,
\end{equation*}
so the metric is the convergence exponent associated with an accumulated squared Bures separation.

The purpose of the present paper is to identify the history theoretic and operational content of this structure. Geometric formulations of consistency already exist, most notably Craig's treatment of a positive decoherence functional as a semi-inner product on history operators \cite{Cr97}. The geometry developed here is different. It is an \emph{asymptotic rate geometry on the record classes of histories}. The basic object is not the value of a decoherence functional at a fixed truncation, but the rate at which cumulative records separate as progressively more registration events are included.

The central result can be summarized schematically as
\begin{equation}\label{eq:three-readings-intro}
\boxed{
\begin{aligned}
\text{sector separation rate}
&=\text{normalized decoherence rate}\\
&=\text{optimal discrimination error rate}.
\end{aligned}
}
\end{equation}
The starting point is deliberately abstract. A history
\begin{equation*}
\hist=(\Omega_1,\Omega_2,\ldots)
\end{equation*}
is assigned, at truncation time $N$, a branch vector whose record component factorizes,
\begin{equation*}
\Psi_{\hist,N} = c_{\hist,N} \bigotimes_{m\le N}r_m(\hist).
\end{equation*}
The normalized Dowker - Halliwell ratio then cancels the branch amplitudes and becomes exactly the cumulative record fidelity. A repeated registration collision model provides a concrete realization, but the geometric results depend only on the sequence of per-step fidelities $F_m(\hist,\histb)$ or, equivalently, the defects
\begin{equation*}
b_m(\hist,\histb) = 1-F_m(\hist,\histb).
\end{equation*}
For spin coherent records these defects are, up to universal multiplicative constants,
\begin{equation*}
b_m \asymp \min\big(1\,,Sd(\Omega_m,\Theta_m)^2\big),
\end{equation*}
providing a direct transfer principle from classical trajectory separation to quantum history geometry.

Section~\ref{sec:bures} places this defect in its natural geometry and proves two sided estimates relating cumulative record fidelity to the accumulated squared Bures distance. In particular, away from perfect single step records, the negative logarithm of cumulative fidelity and the accumulated Bures defect have the same asymptotic scale.

Section~\ref{sec:dichotomy} applies the Kakutani - von Neumann infinite product dichotomy. If
\begin{equation*}
\sum_m b_m<\infty,
\end{equation*}
the two record sequences are weakly equivalent. If
\begin{equation*}
\sum_m b_m=\infty,
\end{equation*}
their infinite product sectors are orthogonal, the corresponding product states are disjoint, and the normalized interference vanishes exactly in the infinite registration limit. A single perfectly orthogonal record can also annihilate interference without producing sector separation; thus exact consistency and sector separation differ on this degenerate stratum. At the level of families, the summability relation defines merger classes, and the quasi-local record representation becomes block diagonal across distinct classes.

Section~\ref{sec:ultrametric} refines this binary alternative by defining the \emph{consistency rate}
\begin{equation*}
\delta(\hist,\histb)=\cE\big(\big\{\frac12d_B\bigl(r_m(\hist),r_m(\histb)\bigr)^2\big\}_{m\ge1}\big),
\end{equation*}
where, for a nonnegative sequence $a=(a_1,a_2,\ldots)$,
\begin{equation}
\cE(a)=\inf\big\{p\ge0:\sum_{m\ge1}a_m\,m^{-p}<\infty\big\}.
\end{equation}
The function $\delta$ is a pseudo-ultrametric on history space and descends to a pseudo-ultrametric on the merger quotient. After quotienting by its full zero distance relation, it becomes a complete ultrametric. The completeness theorem is proved directly by a block pasting construction and requires neither compactness of the outcome space nor continuity of the record family. The same rate geometry also generates a canonical nested filtration of history families by ultrametric balls.

The zero distance relation is strictly coarser than merger equivalence. In particular, there is a \emph{marginal stratum}
\begin{equation*}
\sum_m b_m=\infty,\;\text{ and }\;\delta(\hist,\histb)=0,
\end{equation*}
in which the histories lie in distinct superselection sectors even though their separation has zero polynomial rate. Thus the ultrametric refines the sector dichotomy without reducing to it.

Section~\ref{sec:rate} gives the rate its direct physical interpretation. If the per-step fidelities remain bounded away from zero, then
\begin{equation*}
\delta(\hist,\histb)=\limsup_{N\to\infty}\,\frac{\log^+\bigl(-\log|\mathfrak d_N(\hist,\histb)|\bigr)}{\log N}.
\end{equation*}
For a decohering pair, the exact Helstrom formula for two pure cumulative records further gives
\begin{equation*}
\delta(\hist,\histb)=\limsup_{N\to\infty}\,\frac{\log\bigl(-\log \eP_{\mathrm{err}}^{(N)}\bigr)}{\log N}\,,
\end{equation*}
where $\eP_{\mathrm{err}}^{(N)}$ is the minimum binary discrimination error. Under an additional regularity hypothesis,
\begin{equation*}
|\mathfrak d_N(\hist,\histb)|
=
\exp\bigl(-N^{\delta+o(1)}\bigr).
\end{equation*}
The marginal stratum $\delta=0$ includes slower decay laws, such as polynomial suppression of interference.

Section~\ref{sec:lyapunov} translates classical separation profiles into consistency rate strata. For fixed spin $S$, a power-law profile
\begin{equation*}
d_m^2\asymp m^{-s}
\end{equation*}
gives merger for $s>1$, marginal sector separation with $\delta=0$ at
$s=1$, and
\begin{equation*}
\delta=1-s
\end{equation*}
for $0\le s<1$. A positive long-time mean defect gives the maximal value $\delta=1$. We deliberately separate these metric statements from the model dependent dynamical question of which separation profiles are actually realized by a given autonomous quantum system.

Finally, Section~\ref{sec:pcrr} returns to \cite{PCRR26}. Their permutation sector records are not the spin coherent toy records used here as a geometric example: the coherent labels belong to the collective sector $\cS$, whereas the records themselves live in the permutation sector $\cP$. What is shared is the accumulated fidelity structure. Their finite dimensional record fidelity law provides a concrete finite size counterpart of the product mechanism that the infinite registration model makes exact. We also examine the capacity of a finite dimensional record space. A Welch bound \cite{W74} constrains the number of records in the controlled near orthogonal regime, whereas a random packing construction shows that at a fixed nonzero fidelity threshold finite dimension alone permits exponentially many record states. This is a dimensional non-obstruction, not a claim that the dynamics of \cite{PCRR26} realizes an optimal packing.

We emphasize the division of labor. The dynamical mechanism of \cite{PCRR26}, chaos together with weak symmetry breaking disorder producing an internal bath, is not derived here. Nor is the infinite tensor product asserted to be a literal large system limit of their permutation sector. It instead supplies an asymptotic idealization of continued record accumulation: approximate record separation is promoted to a sharp sector alternative, while the rate of that separation becomes
a complete ultrametric geometry with a direct operational meaning.

\subsection*{Disclosure on the use of AI}
In the course of preparing this manuscript, the author used Claude (Anthropic) and ChatGPT (OpenAI) as research assistants, for verification of proofs and calculations, checking of references, and editorial suggestions on the exposition. All mathematical results and their proofs are due to the author, who has reviewed and verified all AI-assisted material and takes full responsibility for the content of this paper.

\section{The repeated branching model}\label{sec:model}

Let $(X,d)$ be a compact metric space of \emph{outcomes}; in the examples below it is either a finite pointer set or the sphere $\bS^2$ of spin coherent labels of \cite{PCRR26}. A \emph{history} is a sequence $\hist=(\Omega_1,\Omega_2,\dots)\in X^{\infty}$ of outcomes, one per time step.

Each step of the branching process leaves a trace in a fresh \emph{register}, a copy of a fixed Hilbert space $\cH_{\mathrm{reg}}$, indexed by the time step at which it is written. The trace is specified by a \emph{record family}, namely a family of unit vectors $r(\omega)\in\cH_{\mathrm{reg}}$, $\omega\in X$, depending continuously on the outcome. Its basic invariant is the \emph{two-point fidelity}
\begin{equation}\label{eq:two-point}
F(\omega,\omega')\;=\;\bigl|\langle r(\omega),r(\omega')\rangle\bigr|.
\end{equation}
We assume throughout that the family is \emph{nondegenerate}: $F(\omega,\omega)=1$, while $F(\omega,\omega')<1$ for $\omega\neq\omega'$: distinct outcomes leave distinguishable, but in general not orthogonal, traces. A history then accumulates the product record
\begin{equation}\label{eq:record-map}
R(\hist)\;=\;\bigotimes_{m\ge1} r(\Omega_m)\;\in\;\bigotimes_{m\ge1}\cH_{\mathrm{reg}},
\end{equation}
a vector in the complete infinite tensor product of the registers \cite{vN39}. 

For a pair of histories $\hist,\histb$ we write
\begin{equation}\label{eq:per-step}
\begin{split}
F_m(\hist,\histb)&=F(\Omega_m,\Theta_m),\\
b_m(\hist,\histb)&=1-F_m(\hist,\histb),
\end{split}
\end{equation}
for the per-step record fidelity and the per-step \emph{defect}. The sequence $\{b_m\}$ is the only data that the results of this paper use; the sector structure of the infinite tensor product \cite{vN39,LesTP} will decide, in Section~\ref{sec:dichotomy}, when the cumulative records of two histories are macroscopically distinct.

\begin{example}[Record families]\label{ex:records}
(i) \emph{Pointer records} \cite{Z03}. Let $X$ be a finite set of outcomes, and associate with each $\omega\in X$ a unit record state
$r(\omega)\in\cH_{\mathrm{reg}}$, with
\begin{equation*}
F(\omega,\omega')=
\bigl|\langle r(\omega),r(\omega')\rangle\bigr|<1,\qquad \omega\neq\omega'.
\end{equation*}
Thus distinct outcomes leave distinct, though not necessarily orthogonal, pointer records. The simplest example is a binary outcome space $X=\{\uparrow,\downarrow\}$, for which a single qubit can serve as the register. If the two record states are separated by Bloch sphere angle
$2\gamma$, then
\begin{equation*}
F(\uparrow,\downarrow)=|\cos\gamma|.
\end{equation*}

\noindent
(ii) \emph{Spin coherent records}. Let $X=\bS^2$ with the geodesic (great circle) metric: writing $\hat n_\omega\in\bS^2$ for the unit vector with polar coordinates $\omega=(\theta,\varphi)$,
\begin{equation}\label{eq:geodesic}
\begin{split}
d(\omega,\omega')&=\arccos\bigl(\hat n_\omega\cdot\hat n_{\omega'}\bigr)\\
&=\arccos\bigl(\cos\theta\cos\theta'+\sin\theta\sin\theta'\cos(\varphi-\varphi')\bigr)\in[0,\pi].
\end{split}
\end{equation}
Equivalently, $\|\hat n_\omega-\hat n_{\omega'}\|=2\sin\bigl(d(\omega,\omega')/2\bigr)$. Let $\cH_{\mathrm{reg}}=\bC^{2S+1}$, and let $r(\omega)=\ket{\omega}$ be the spin $S$ coherent state. Since $\ket{\omega}$ is the $2S$-fold symmetric tensor power of the spin $\frac12$ coherent state
\begin{equation*}
\ket{\omega}_{1/2}=\cos\frac\theta2\,\ket{0}+e^{i\varphi}\sin\frac\theta2\,\ket{1},
\end{equation*}
the overlaps satisfy $\langle\omega,\omega'\rangle=\bigl(\langle\omega,\omega'\rangle_{1/2}\bigr)^{2S}$, where
\begin{equation*}
\langle\omega,\omega'\rangle_{1/2}=\cos\frac\theta2\,\cos\frac{\theta'}2+e^{i(\varphi'-\varphi)}\sin\frac\theta2\,\sin\frac{\theta'}2\,,
\end{equation*}
with
\begin{equation*}
\begin{split}
\bigl|\langle\omega,\omega'\rangle_{1/2}\bigr|^2&=\frac12\bigl(1+\hat n_\omega\cdot\hat n_{\omega'}\bigr)\\
&=\cos^2\bigl(d(\omega,\omega')/2\bigr),
\end{split}
\end{equation*}
by a direct expansion. Hence
\begin{equation*}
F(\omega,\omega')=\cos^{2S}\bigl(d(\omega,\omega')/2\bigr).
\end{equation*}
In particular, the family is nondegenerate, $F(\omega,\omega')<1$ for $\omega\neq\omega'$, and $F(\omega,\omega')=0$ precisely at antipodal pairs, $d(\omega,\omega')=\pi$.

The following elementary estimate will be the bridge between classical separation of coherent-state labels and quantum separation of their records.

\begin{lemma}[Cosine-power estimate]\label{lem:cosine-power}
For $p\ge1$ and $|x|\le\pi/2$,
\begin{equation}\label{eq:cospx}
(1-e^{-1})\min\Big(1,\;\frac{px^2}{2}\Big)\le1-(\cos x)^p\le\min\Big(1,\;\frac{px^2}{2}\Big).
\end{equation}
Both constants are optimal.
\end{lemma}

The proof is elementary and is given in Appendix~\ref{app:cosine-power}. Applying Lemma~\ref{lem:cosine-power} with
\begin{equation*}
\begin{split}
x&=\frac{d(\omega,\omega')}{2},\\
p&=2S,
\end{split}
\end{equation*}
shows in particular that
\begin{equation}\label{eq:coherent-defect}
1-F(\omega,\omega')
\asymp
\min\big(1,\,Sd(\omega,\omega')^2\big).
\end{equation}
More precisely, one has the sharp global estimate
\begin{equation}\label{eq:coherent-defect-sharp}
(1-e^{-1/4})\min\big(1,\,Sd(\omega,\omega')^2\big)\le 1-F(\omega,\omega')\le\min\big(1,\,Sd(\omega,\omega')^2\big).
\end{equation}
Indeed, putting $a=S\,d(\omega,\omega')^2$ and using
$\cos x\le e^{-x^2/2}$ for $|x|\le\pi/2$ (established in Appendix~\ref{app:cosine-power}) gives
\begin{equation*}
1-F(\omega,\omega')\ge 1-e^{-a/4}.
\end{equation*}
For $0<a\le1$, the function $(1-e^{-a/4})/a$ is decreasing, while for $a\ge1$ the function $1-e^{-a/4}$ is increasing in $a$; hence
\begin{equation*}
1-F(\omega,\omega') \ge (1-e^{-1/4})\min(1,\,a).
\end{equation*}
The upper bound follows from
\begin{equation*}
1-\cos^{2S}(d/2)\le \frac{S d^2}{4}\le S d^2
\end{equation*}
together with the trivial bound by $1$. The lower constant is sharp along $S\to\infty$, $d=S^{-1/2}$: then $a=1$, while $F=\cos^{2S}\bigl(\tfrac12 S^{-1/2}\bigr)\to e^{-1/4}$. The upper constant is sharp at antipodal
pairs, where $1-F=1=\min(1,\,Sd^2)$.

Equation~\eqref{eq:coherent-defect} is the transfer principle from classical trajectory separation to quantum history geometry: up to fixed multiplicative constants, it identifies the per-step record defect with the truncated squared distance between the corresponding classical labels. It will therefore allow the asymptotic separation profiles considered below to be translated directly into consistency rate strata.

Finally, the small-angle behavior follows from
$\log\cos u=-u^2/2+O(u^4)$:
\begin{equation}\label{eq:coherent-small-angle}
F(\omega,\omega')=\exp\Big(-\frac{S\,d(\omega,\omega')^2}{4}\big(1+O(d(\omega,\omega')^2)\big)\Big),
\end{equation}
as $d(\omega,\omega')\to0$, where the $O(d^2)$ term inside the exponent is uniform in $S$.

This spin coherent family is a convenient geometric toy model. It should not be identified literally with the records of \cite{PCRR26}: there the spin coherent state belongs to the collective sector $\cS$, whereas the microscopic record belongs to the permutation sector $\cP$. The connection is instead through the accumulated record fidelity law discussed in Remark~\ref{rem:pcrr-match} and Section~\ref{sec:pcrr}.
\end{example}

Everything in this paper rests on a single structural assumption about branches, which we now isolate.

\begin{definition}[Branch structure]\label{def:branch}
A \emph{branch structure} over a family of histories assigns to each history $\hist$ and each $N\ge1$ a \emph{branch vector}
\begin{equation}\label{eq:branch-vec}
\begin{split}
\Psi_{\hist,N}&=c_{\hist,N}\,R_N(\hist),\\
R_N(\hist)&=\bigotimes_{m\le N} r(\Omega_m),
\end{split}
\end{equation}
with nonzero amplitudes $c_{\hist,N}\in\bC$: the record part factorizes over the time steps. The \emph{decoherence functional} of a pair of histories, truncated at $N$ steps, is the Gell-Mann - Hartle bilinear form \cite{GMH90}
\begin{equation}\label{eq:dec-functional}
\cD_N(\hist,\histb)\;=\;\bigl\langle\Psi_{\histb,N},\,\Psi_{\hist,N}\bigr\rangle,
\end{equation}
so that, by \eqref{eq:branch-vec},
\begin{equation}\label{eq:offdiag}
\cD_N(\hist,\histb)\;=\;\bar c_{\histb,N}\,c_{\hist,N}\prod_{m\le N} \bigl\langle r(\Theta_m),r(\Omega_m)\bigr\rangle .
\end{equation}
\end{definition}

The raw off-diagonal \eqref{eq:offdiag} is not the right object for grading consistency: it can become small simply because the norms of the individual branch vectors become small. This occurs, for example, in product branching models whose path weights decay exponentially with the number of branchings. Raw smallness therefore conflates suppression of interference with smallness of the branches themselves. The remedy is the normalized consistency criterion of Dowker and Halliwell \cite{DH92},
\begin{equation}\label{eq:DH}
\begin{split}
\mathfrak d_N(\hist,\histb)&=\frac{\cD_N(\hist,\histb)}{\sqrt{\cD_N(\hist,\hist)\,\cD_N(\histb,\histb)}}\,,\\
\bigl|\mathfrak d_N(\hist,\histb)\bigr|&=\;\prod_{m\le N}F_m(\hist,\histb),
\end{split}
\end{equation}
in which the amplitudes cancel exactly. The Dowker - Halliwell ratio of a pair of histories \emph{is} the cumulative fidelity of their records. Every result of this paper depends on a branch structure only through the per-step defect sequence $\{b_m\}$; all limit and rate statements below concern \eqref{eq:DH}.

It remains to show that branch structures of this form can arise dynamically. The canonical realization is a collision model with one fresh register per time step, so that records never crowd. This is the idealization in which the finite internal bath $\cP$ of \cite{PCRR26} is replaced by an indefinitely extendible sequence of record degrees of freedom.

A word on the formalism is useful. The class operators below are not chains of orthogonal projections on a fixed Hilbert space, as in the projective consistent histories framework \cite{Gri84,Omn92,GMH90}, but Kraus components of a repeated quantum instrument acting between growing register spaces. The construction is therefore a \emph{generalized} histories model in the sense of Isham \cite{I94}. The use of generalized histories is essential here: if every outcome were retained in an orthogonal persistent register, distinct histories would be orthogonal from the outset and the consistency problem would become trivial; see Remark~\ref{rem:trap}.

\begin{definition}[Repeated registration]\label{def:iterated}
Let $X$ be a finite alphabet. At this abstract level the transient system degrees of freedom have been integrated out: the step operators retain only the branch amplitude and the record
deposited in the new register. For each step $m$ fix amplitudes $a^{(m)}\colon X\to\bC$ with $\sum_{\omega}|a^{(m)}_\omega|^2=1$, and define the step-$m$ \emph{branch operator} 
\begin{equation*}
A^{(m)}_\omega\colon\bigotimes_{n<m}\cH_{\mathrm{reg}}\to\bigotimes_{n\le m}\cH_{\mathrm{reg}},
\end{equation*}
by
\begin{equation}\label{eq:branch-op}
A^{(m)}_\omega\,\Psi\;=\;a^{(m)}_\omega\,\Psi\otimes r(\omega),
\end{equation}
which satisfies the per-step completeness relation
\begin{equation}
\sum_\omega A^{(m)\dagger}_\omega A^{(m)}_\omega=1.
\end{equation}
The \emph{class operator} of a history truncated at $N$ steps is the chain
\begin{equation}
C_{\hist,N}=A^{(N)}_{\Omega_N}\cdots A^{(1)}_{\Omega_1},
\end{equation}
so that, applied to the empty register line (the unit $1$ of the empty tensor product),
\begin{equation}\label{eq:class-branch}
\begin{split}
C_{\hist,N}\,1&=c_{\hist,N}\;r(\Omega_1)\otimes\cdots\otimes r(\Omega_N),\\
c_{\hist,N}&=\prod_{m\le N}a^{(m)}_{\Omega_m}.
\end{split}
\end{equation}
The weights $\mu_N(\hist)=|c_{\hist,N}|^2$ form a product probability measure on $X^N$, the classical path measure of the branching process.
\end{definition}

\begin{lemma}\label{lem:offdiag}
Applied to the empty register line, the chains of Definition~\ref{def:iterated} realize the branch structure of Definition~\ref{def:branch}: $C_{\hist,N}\,1=\Psi_{\hist,N}$ with
$c_{\hist,N}=\prod_{m\le N}a^{(m)}_{\Omega_m}$, and
$\cD_N(\hist,\histb)=\langle C_{\histb,N}1,\,C_{\hist,N}1\rangle$.
\end{lemma}

\begin{proof}
Immediate from \eqref{eq:class-branch} and the product structure of the registers.
\end{proof}

\begin{example}[One branching step as a controlled unitary]\label{ex:minimal}
A single step of Definition~\ref{def:iterated} is implemented unitarily as follows. Take $X=\{\uparrow,\downarrow\}$, and couple a transient \emph{probe} qubit
$\cH_{\mathrm{probe}}=\bC^2$, prepared in the superposition
\begin{equation*}
\chi^{(m)}=a^{(m)}_\uparrow\ket\uparrow+a^{(m)}_\downarrow\ket\downarrow,
\end{equation*} to the fresh register $m$, prepared in a ready state $\varphi\in\cH_{\mathrm{reg}}$, through the controlled unitary 
\begin{equation}\label{branchV}
V^{(m)}\;=\;P_\uparrow\otimes W_\uparrow\,+\,P_\downarrow\otimes W_\downarrow.
\end{equation}
Here $P_\uparrow$, $P_\downarrow$ are the projections onto the probe basis and $W_\uparrow$, $W_\downarrow$ are unitaries of $\cH_{\mathrm{reg}}$ preparing the records,
$r(\omega)=W_\omega\varphi$. One verifies immediately that $V^{(m)}$ is unitary, and
\begin{equation}\label{branchEq}
V^{(m)}\bigl(\chi^{(m)}\otimes\varphi\bigr)
\;=\;a^{(m)}_\uparrow\,\ket\uparrow\otimes r(\uparrow)
\,+\,a^{(m)}_\downarrow\,\ket\downarrow\otimes r(\downarrow).
\end{equation}
Thus the wavefunction branches, with each branch depositing its corresponding record. The branch operators \eqref{eq:branch-op} are the probe matrix elements of this event,
\begin{equation}\label{eq:probe-element}
\begin{split}
A^{(m)}_\omega\,\Psi
&=\Psi\otimes\bigl(\bra{\omega}\,V^{(m)}\,\ket{\chi^{(m)}}\bigr)\varphi\\
&=a^{(m)}_\omega\,\Psi\otimes r(\omega),
\end{split}
\end{equation}
and so Definition~\ref{def:iterated} is Example~\ref{ex:minimal} iterated with a fresh probe and a fresh register at every step, the probes carrying the amplitudes.

The one shot branching event of \cite{LesTP}, Section~6.2,  is the variant in which a single persistent probe controls product unitaries $\mathbf U=\bigotimes_j U_j$ and $\mathbf U'=\bigotimes_j U'_j$ acting on the \emph{entire} register line at once. The two branches are then the product records $(U_j\varphi_j)_j$ and $(U'_j\varphi_j)_j$, and whether they generate the same incomplete tensor product (one world) or orthogonal ones (two worlds) is decided by a single summability condition, the prototype of the dichotomy of Section~\ref{sec:dichotomy}. Definition~\ref{def:iterated} serializes that event: the writing is localized to one register per step, and time, rather than the site index, enumerates the environment.
\end{example}

\begin{remark}[Why no orthogonal outcome register is retained]\label{rem:trap} The probes of Example~\ref{ex:minimal} are transient: the class operators sum them out, and no register stores the outcome orthogonally. This is forced, not a convenience. If each outcome were additionally stored in an orthonormal register that persists  (iterate Example~\ref{ex:minimal} keeping every probe), then distinct histories would be exactly orthogonal through their probes alone, the decoherence functional identically diagonal, and consistency trivial. The nontrivial regime, and the one matching \cite{PCRR26}, whose coherent state projectors are non-orthogonal, is the one in which the only trace a history leaves is its record, with per-step overlap $F<1$.
\end{remark}

\begin{remark}[Windowed records]\label{rem:window}
Nothing below requires the special form $r_m=r(\Omega_m)$, only the defect sequence $\{b_m\}$ enters. In particular, the record at step $m$ may depend on a bounded window of the history, for example $r_m=r(\Omega_{m-1},\Omega_m)$. This is the form naturally suggested by the accumulated functional in \cite{PCRR26}, which contains both endpoints of each time interval. All pairwise sector, rate, and discrimination results remain unchanged after replacing $b_m$ by the defect of the corresponding windowed records. The only additional issue is completeness under block pasting, handled in Proposition~\ref{prop:window-complete}.

The collision model realization of Definition~\ref{def:iterated} is stated for a finite alphabet so that the instrument normalization is an ordinary sum. The abstract branch structure formulation, and hence all subsequent geometric results, does not require a finite outcome set. Thus the spin coherent example on $\bS^2$ is used at the level of record geometry without introducing a coherent state resolution of the identity into the collision model.
\end{remark}

\section{Records and the accumulated Bures functional}\label{sec:bures}

Throughout this section and the next, $\hist,\histb$ is a fixed pair of histories, with per-step fidelities $F_m=F(\Omega_m,\Theta_m)$ and defects $b_m=1-F_m$ as in
\eqref{eq:per-step}.

We first place the defect in its natural geometry. For unit vectors, the two-point fidelity \eqref{eq:two-point} is the fidelity of the corresponding
pure states. For pure states Uhlmann's fidelity \cite{U76,J94} 
\begin{equation*}
F(\rho,\sigma)=\tr\bigl|\sqrt{\rho}\sqrt{\sigma}\bigr|
\end{equation*}
reduces to $|\langle r,r'\rangle|$. With this convention, the associated Bures distance \cite{Bu69} satisfies
\begin{equation}\label{eq:bures-def}
d_B(r,r')^2=2\bigl(1-F(r,r')\bigr).
\end{equation}
Thus the per-step defect is exactly half the squared Bures distance,
\begin{equation*}
b_m = \frac12\,d_B\bigl(r(\Omega_m),r(\Theta_m)\bigr)^2.
\end{equation*}
This is not an additional approximation: the defect entering the infinite product criterion below is precisely the squared Bures separation of the corresponding pure records, up to the fixed factor $1/2$. Accordingly, the convergence exponents studied here are Bures divergence exponents in the sense of Equations~(52) - (53) of \cite{LesTP}.

The cumulative record fidelity, by \eqref{eq:DH}, the modulus of the Dowker - Halliwell ratio, is controlled two-sidedly by the accumulated defect.

\begin{proposition}[Master estimate]\label{prop:master}
Let
\begin{equation*}
\begin{split}
\Sigma_N&=\sum_{m\le N}b_m,\\
\beta_N&=\max_{m\le N}b_m.
\end{split}
\end{equation*}
Then:
\begin{itemize}
\item[(i)] Unconditionally,
\begin{equation*}
\prod_{m\le N}F_m\le e^{-\Sigma_N}.
\end{equation*}

\item[(ii)] If $F_m\ge c>0$ for all $m\le N$, then
\begin{equation*}
\prod_{m\le N}F_m \ge e^{-\kappa(c)\Sigma_N},
\end{equation*}
where
\begin{equation*}
\kappa(c)=\frac{-\log c}{1-c}.
\end{equation*}

\item[(iii)] If $\beta_N<1$, then
\begin{equation}\label{eq:log-identity}
\begin{split}
\Sigma_N &\le -\log\prod_{m\le N}F_m\\
&\le \frac{\Sigma_N}{1-\beta_N}\,.
\end{split}
\end{equation}
\end{itemize}

In particular, if $\beta_N$ is small, then
\begin{equation}\label{eq:bures-accumulated}
-\log\prod_{m\le N}F_m =\frac12\sum_{m\le N} d_B\bigl(r(\Omega_m),r(\Theta_m)\bigr)^2 \bigl(1+O(\beta_N)\bigr).
\end{equation}
Thus, in the small defect regime, the negative logarithm of cumulative record fidelity is, to relative error $O(\beta_N)$, the accumulated squared Bures distance of the per-step records.
\end{proposition}

\begin{proof}
All three statements follow from the integral representation
\begin{equation*}
-\log t=\int_t^1\frac{ds}{s}\,, \qquad t\in(0,1].
\end{equation*}

\noindent
(i) If some $F_m=0$, the assertion is immediate. Otherwise, $s^{-1}\ge1$ on $[t,1]$ gives
\begin{equation*}
-\log t\ge1-t.
\end{equation*}
Applying this to $t=F_m$ and summing over $m$ yields
\begin{equation*}
\begin{split}
-\log\prod_{m\le N}F_m&=\sum_{m\le N}(-\log F_m)\\
&\ge \sum_{m\le N}(1-F_m)\\
&=
\Sigma_N,
\end{split}
\end{equation*}
which is equivalent to the stated bound.

\noindent
(ii) Define
\begin{equation*}
\begin{split}
g(t)
&=\frac{-\log t}{1-t}\\
&=\frac{1}{1-t}
\int_0^{1-t}\frac{du}{1-u}.
\end{split}
\end{equation*}
This is the average of the increasing function $u\mapsto(1-u)^{-1}$ over the interval $[0,1-t]$, and is therefore decreasing in $t$. Hence, for $F_m\in[c,1]$,
\begin{equation*}
\begin{split}
-\log F_m &\le g(c)(1-F_m)\\
&=\kappa(c)b_m.
\end{split}
\end{equation*}
Summing over $m$ gives
\begin{equation*}
-\log\prod_{m\le N}F_m
\le
\kappa(c)\Sigma_N,
\end{equation*}
and exponentiation proves the claim.

\noindent
(iii) The lower bound is part (i). For the upper bound,
\begin{equation*}
\begin{split}
-\log(1-b_m)
&=\int_0^{b_m}\frac{du}{1-u}\\
&\le\frac{b_m}{1-b_m}\\
&\le\frac{b_m}{1-\beta_N}.
\end{split}
\end{equation*}
Summing over $m$ gives the upper bound in
\eqref{eq:log-identity}. The final assertion follows from
\begin{equation*}
\frac{1}{1-\beta_N}=1+O(\beta_N)
\end{equation*}
together with \eqref{eq:bures-def}.
\end{proof}

\begin{remark}[Correspondence with the record-fidelity law of \cite{PCRR26}]\label{rem:pcrr-match} 
The spin coherent family of Example~\ref{ex:records}(ii) illustrates how an accumulated squared classical distance can arise from an elementary product of quantum overlaps:
\begin{equation*}
b_m \asymp \min\big(1,\,Sd(\Omega_m,\Theta_m)^2\big).
\end{equation*}
The records in \cite{PCRR26}, however, are not these coherent states. Their coherent labels belong to the collective sector $\cS$, whereas the record vectors belong to the permutation sector $\cP$. What matches the present construction is the \emph{functional form} of their accumulated permutation record fidelity:
\begin{equation*}
|\langle R_{\histb}|R_{\hist}\rangle|^2 \sim \exp\bigl(-\Gamma D_{\hist,\histb}^{(2)}\bigr),
\end{equation*}
where
\begin{equation*}
D_{\hist,\histb}^{(2)}\approx\frac12\sum_{m=1}^{M}\big(d(\Omega_{m-1},\Theta_{m-1})^2+d(\Omega_m,\Theta_m)^2\big).
\end{equation*}
Since our fidelity convention uses the unsquared overlap, the corresponding relation is
\begin{equation*}
|\langle R_{\histb}|R_{\hist}\rangle|\sim \exp\bigl(-\frac12\,\Gamma D_{\hist,\histb}^{(2)}\bigr).
\end{equation*}
The factor $1/2$ is immaterial for all convergence and rate exponents considered below. Thus the finite bath of \cite{PCRR26} realizes the same accumulated fidelity principle that the fresh register model makes exact by construction. Matching the coefficient
\begin{equation*}
\Gamma=\frac{\gamma S^2\Delta t}{1+\gamma S\Delta t}
\end{equation*}
requires the semiclassical Lindblad analysis of \cite{PCRR26} and is not derived here.
\end{remark}

Proposition~\ref{prop:master} therefore reduces the asymptotic behavior of normalized interference to that of the additive functional $\Sigma_N$. The infinite product dichotomy and its rate refinement developed below are consequences of this reduction.

\section{Exact consistency and sector separation: the Kakutani - von Neumann dichotomy}\label{sec:dichotomy}

In the infinite registration limit, the gradual suppression of interference becomes an exact dichotomy. Two cumulative record sequences either remain in the same weak equivalence class or separate into disjoint infinite product sectors. In the absence of a perfect record, this alternative is equivalently expressed by survival or vanishing of the accumulated fidelity. Thus the approximate record separation observed at finite size in \cite{PCRR26} acquires an exact asymptotic counterpart.

The mechanism is the dichotomy governing infinite products - -of overlaps on the Hilbert space side and of measures on the probabilistic side - -which we now state formally. Throughout,
\begin{equation*}
\cH_{\mathrm{reg}}^{\otimes\infty}=\bigotimes_{m\ge1}\cH_{\mathrm{reg}}
\end{equation*}
denotes the complete infinite tensor product of the registers \cite{vN39}, and
\begin{equation*}
\mathfrak A = \bigotimes_{m\ge1}\mathcal B(\cH_{\mathrm{reg}})
\end{equation*}
denotes the quasi-local register algebra, i.e. the norm closure of the operators acting on finitely many registers.

\subsection{Kakutani - von Neumann dichotomy}
\label{subsec:KvN}

\begin{definition}[Equivalence of record sequences]\label{def:sectors}
Let $R=(r_m)_{m\ge1}$ and $R'=(r'_m)_{m\ge1}$ be sequences of unit vectors in $\cH_{\mathrm{reg}}$. They are
\begin{itemize}
\item[(a)] \emph{equivalent}, $R\sim R'$, if
\begin{equation*}
\sum_m\bigl|1-\langle r_m,r'_m\rangle\bigr|<\infty;
\end{equation*}
\item[(b)] \emph{weakly equivalent}, $R\sim_w R'$, if
\begin{equation*}
\sum_m\bigl(1-|\langle r_m,r'_m\rangle|\bigr)<\infty.
\end{equation*}
\end{itemize}
Thus weak equivalence is equivalence up to a choice of phase in each register, see \cite{LesTP}, Section 4. The \emph{incomplete tensor product} generated by $R$ is the closed subspace
\begin{equation*}
\cH_R = \overline{\mathrm{span}}\,\Bigl\{{\textstyle\bigotimes_m}r'_m:R'\sim R \Bigr\} \subset \cH_{\mathrm{reg}}^{\otimes\infty},
\end{equation*}
the \emph{sector} of $R$.
\end{definition}

\begin{theorem}[Kakutani - von Neumann dichotomy \cite{vN39,K48}]\label{thm:KvN}
For any two sequences $R,R'$ of unit vectors, exactly one of the following holds.
\begin{itemize}
\item[(i)] $R\sim_w R'$. Then
\begin{equation*}
\prod_m|\langle r_m,r'_m\rangle|
\end{equation*}
converges, with a strictly positive limit if and only if $\langle r_m,r'_m\rangle\neq0$ for every $m$. The phases in the individual registers can be chosen so that $R\sim R'$, in which case $\cH_R=\cH_{R'}$. The product states induced on the quasi-local algebra $\mathfrak A$ are quasi-equivalent \cite{P67,AW68}.

\item[(ii)] $R\not\sim_w R'$. Then
\begin{equation*}
\prod_m|\langle r_m,r'_m\rangle|=0,
\end{equation*}
the sectors are orthogonal, $\cH_R\perp\cH_{R'}$, for every choice of phases, and the induced product states are disjoint. Moreover, no quasi-local observable connects the sectors:
\begin{equation}\label{eq:superselection}
\bigl\langle\Phi,A\Psi\bigr\rangle=0
\qquad
\text{for all }
A\in\mathfrak A,\;
\Phi\in\cH_R,\;
\Psi\in\cH_{R'}.
\end{equation}
\end{itemize}
\end{theorem}

The convergence statements are elementary: the partial products are nonincreasing, and the usual infinite product criterion applies. The sector and superselection statements are von Neumann's \cite{vN39}, \S\S4 - 5, in the gauge invariant form of Section~4 and Proposition~6.1 of \cite{LesTP}; we do not reprove them here.

\begin{remark}[The classical dichotomy]\label{rem:kakutani-classical}
Theorem~\ref{thm:KvN} is the Hilbert space counterpart of Kakutani's theorem on infinite product measures \cite{K48}. If
\begin{equation*}
\begin{split}
\mu&=\bigotimes_m\mu_m,\\
\nu&=\bigotimes_m\nu_m
\end{split}
\end{equation*}
are product probability measures with $\mu_m\sim\nu_m$ for every $m$, then $\mu$ and $\nu$ are either equivalent or mutually singular, according as the product of the per factor Hellinger affinities
\begin{equation*}
\rho_m=\int\sqrt{d\mu_m\,d\nu_m}
\end{equation*}
is positive or zero. Equivalently, according as whether the sum $\sum_m(1-\rho_m)$ converges or diverges. The dictionary with the Hilbert space statement is direct: for pure states, the overlap fidelity is the quantum analogue of the Hellinger affinity, and the two dichotomies are governed by the same product/summability mechanism.

Historically, the order is the reverse of the naming: von Neumann's theory of infinite tensor products appeared in 1939, whereas Kakutani's product measure theorem appeared in 1948. Kakutani's introduction credits von Neumann for the Hilbert space embedding underlying his proof. Bures's subsequent extension of Kakutani's theorem to tensor products of von Neumann algebras \cite{Bu69} is also the source of the metric used in Section~\ref{sec:bures}.
\end{remark}

Applied to the record sequences of a pair of histories, the dichotomy turns asymptotic consistency into an exact statement.

\begin{theorem}[Consistency dichotomy]\label{thm:dichotomy}
For every pair of histories $\hist,\histb$, exactly one of the following holds.
\begin{itemize}
\item[(i)] If
\begin{equation*}
\sum_m b_m<\infty,
\end{equation*}
then the records $R(\hist)$ and $R(\histb)$ are weakly equivalent: the two histories belong to the same weak equivalence class of record sequences. The Dowker - Halliwell ratio converges,
\begin{equation*}
\lim_{N\to\infty}|\mathfrak d_N(\hist,\histb)|=\prod_m F_m,
\end{equation*}
and the limit is strictly positive (so that interference persists at every scale) if and only if $F_m>0$ for every $m$.

\item[(ii)] If
\begin{equation*}
\sum_m b_m=\infty,
\end{equation*}
then
\begin{equation*}
\lim_{N\to\infty}\mathfrak d_N(\hist,\histb)=0.
\end{equation*}
The pair is exactly consistent, and its records lie in orthogonal superselection sectors. In particular, by \eqref{eq:superselection}, no quasi-local observable connects the two branches.
\end{itemize}
\end{theorem}

\begin{proof}
Immediate from Theorem~\ref{thm:KvN} and
\begin{equation*}
|\mathfrak d_N(\hist,\histb)| = \prod_{m\le N}F_m
\end{equation*}
from \eqref{eq:DH}.
\end{proof}

\begin{remark}[Perfect records: consistency without sector separation]
\label{rem:perfect}
The positivity clause in Theorem~\ref{thm:dichotomy}(i) is not a pedantry. A single step at which the pair is recorded orthogonally, $F_{m_0}=0$, annihilates the Dowker - Halliwell ratio for all $N\ge m_0$: consistency is achieved kinematically, at one stroke. Yet $\sum_m b_m$ may still converge, in which case the two record sequences remain weakly equivalent and no asymptotic sector separation occurs.

Thus consistency and sector separation are logically distinct:
\begin{equation*}
\sum_m b_m=\infty
\end{equation*}
implies sector separation and exact consistency, whereas
\begin{equation*}
F_{m_0}=0
\end{equation*}
implies exact consistency, but not necessarily sector separation. The two criteria agree away from this degenerate stratum, which is precisely where the textbook idealization of measurement by orthogonal pointer states resides; cf.\ Remark~\ref{rem:trap}. The rate theory of the following sections excludes this case through the hypothesis $\inf_m F_m>0$ along the pair.
\end{remark}

\subsection{Sector decomposition of history families}
\label{subsec:families}

The dichotomy so far concerns pairs. The consistent histories formalism, however, is about \emph{families} of histories, i.e. exhaustive sets of coarse-grained alternatives. It is at the family level that the infinite registration limit delivers its structural payoff. Throughout, $\mathscr F\subseteq X^\infty$ is an arbitrary family of histories, and we define
\begin{equation}\label{eq:merger}
\hist\sim\histb\;\text{ if and only if }\;\sum_m b_m(\hist,\histb)<\infty
\end{equation}
to be the \emph{merger relation} associated with Theorem~\ref{thm:dichotomy}(i).

\begin{lemma}[Merger is an equivalence relation]\label{lem:merger}
The per-step defect is quasi-subadditive: for any three pure record states $\rho,\sigma,\tau$,
\begin{equation}\label{eq:quasi}
1-F(\rho,\tau)\le 2\Bigl((1-F(\rho,\sigma))+(1-F(\sigma,\tau))\Bigr).
\end{equation}
Consequently, the relation \eqref{eq:merger} is an equivalence relation on $X^\infty$.
\end{lemma}

\begin{proof}
The Bures distance \eqref{eq:bures-def} is a metric on pure-state rays. Equivalently, for unit vector representatives it is the minimum of $\|e^{i\alpha}r-r'\|$ over the phase $\alpha$ \cite{Bu69}. The triangle inequality together with $(a+b)^2\le2(a^2+b^2)$ therefore gives
\begin{equation*}
d_B(\rho,\tau)^2 \le 2\Bigl(d_B(\rho,\sigma)^2+d_B(\sigma,\tau)^2\Bigr),
\end{equation*}
which is \eqref{eq:quasi} after dividing by $2$.

Reflexivity and symmetry of \eqref{eq:merger} are immediate. For transitivity, if $\hist,\histb,\boldsymbol\Xi$ are three histories, then
\begin{equation*}
b_m(\hist,\boldsymbol\Xi) \le 2\Bigl(b_m(\hist,\histb)+b_m(\histb,\boldsymbol\Xi)\Bigr).
\end{equation*}
Hence summability of the two terms on the right implies summability of the term on the left.
\end{proof}

\begin{theorem}[Sector decomposition of history families]
\label{thm:family}
Let $\mathscr F\subseteq X^\infty$ be an arbitrary family of histories, and for a merger class $C\in\mathscr F/\!\sim$ let
\begin{equation*}
\cH_C = \overline{\mathrm{span}}\,\bigcup_{\hist\in C}\cH_{R(\hist)}\subset\cH_{\mathrm{reg}}^{\otimes\infty}
\end{equation*}
be the closed span of the record sectors of its members. Then:
\begin{itemize}
\item[(i)] Distinct merger classes have mutually orthogonal, $\mathfrak A$-invariant record subspaces:
\begin{equation*}
\cH_C\perp\cH_{C'}\qquad(C\neq C'),
\end{equation*}
and
\begin{equation*}
\langle\Phi,A\Psi\rangle=0
\end{equation*}
for every $A\in\mathfrak A$, $\Phi\in\cH_C$, and $\Psi\in\cH_{C'}$.

\item[(ii)] Consequently, the record representation of the quasi-local algebra on
\begin{equation*}
\cH_{\mathscr F}=\overline{\bigoplus_C\cH_C}
\end{equation*}
is block-diagonal over the merger classes:
\begin{equation*}
\pi=\bigoplus_C\pi_C.
\end{equation*}

\item[(iii)] The quotient family is pairwise exactly consistent: for every cross-class pair $[\hist]\neq[\histb]$,
\begin{equation*}
\lim_{N\to\infty}\mathfrak d_N(\hist,\histb)=0.
\end{equation*}
Moreover, the product states induced on $\mathfrak A$ by the cumulative records are quasi-equivalent within each merger class and mutually disjoint across distinct classes.
\end{itemize}
\end{theorem}

\begin{proof}
(i) Let $\hist\in C$ and $\histb\in C'$ with $C\neq C'$. By Theorem~\ref{thm:dichotomy},
\begin{equation*}
\sum_m b_m(\hist,\histb)=\infty,
\end{equation*}
so the record sequences are not weakly equivalent. In particular, no representative of the sector of $R(\hist)$ is equivalent to one of $R(\histb)$. Theorem~\ref{thm:KvN}(ii) therefore gives
\begin{equation*}
\cH_{R(\hist)}\perp\cH_{R(\histb)}
\end{equation*}
together with the superselection identity \eqref{eq:superselection} for that pair of sectors. Both statements pass to closed spans by linearity and continuity. Local operators preserve each sector by modifying only finitely many registers of a product vector; invariance then extends to the quasi-local norm closure by continuity, cf. Proposition~6.1 of \cite{LesTP}.

\noindent
(ii) This is (i) restated at the level of the representation: $\pi(A)$ preserves each $\cH_C$ and has no matrix elements between distinct merger classes.

\noindent
(iii) The first assertion is Theorem~\ref{thm:dichotomy}(ii) applied to each cross-class pair. The second is the quasi-equivalence/disjointness dichotomy of Theorem~\ref{thm:KvN}, applied within and across the merger classes.
\end{proof}

\begin{remark}[Additivity for finite coarse-grainings]
\label{rem:additive}
The sector decomposition supplies the algebraic prerequisite for the usual probability calculus. For every \emph{finite} coarse-graining of the merger quotient, all cross-class interference terms vanish in the infinite registration limit, and so the corresponding limiting weights add exactly. For an arbitrary countable or uncountable family, promoting these finite additivity statements to a $\sigma$-additive probability measure on the quotient requires the usual additional measure theoretic hypotheses; we do not impose them here. At finite $N$, the additivity defect is controlled by the corresponding normalized off-diagonal terms, whose decay is graded in Section~\ref{sec:rate}.
\end{remark}

\section{The ultrametric on the space of histories}\label{sec:ultrametric}

The dichotomy is binary; the geometry lies in its rate.

\begin{definition}[Consistency rate]\label{def:delta}
For histories $\hist,\histb$ set
\begin{equation}\label{eq:delta-def}
\begin{split}
\delta(\hist,\histb)&=\cE\bigl(\{b_m(\hist,\histb)\}_m\bigr)\\
&=\cE\bigl(\bigl\{\tfrac12\,d_B\bigl(r(\Omega_m),r(\Theta_m)\bigr)^2\bigr\}_m\bigr).
\end{split}
\end{equation}
\end{definition}

\begin{theorem}[The history ultrametric]\label{thm:ultrametric}
The function $\delta$ is a pseudo-ultrametric on $X^\infty$, and
\begin{equation}\label{eq:delta-is-td}
\delta(\hist,\histb)=\tilde d\bigl([R(\hist)]_w,[R(\histb)]_w\bigr),\end{equation}
the gauge invariant sector metric of \cite{LesTP} evaluated on the record classes.
\end{theorem}

\begin{proof}
Symmetry and $\delta(\hist,\hist)=0$ are immediate. For the identity \eqref{eq:delta-is-td}, the per-step records are unit vectors, so the fidelity of the associated pure states is the overlap $F_m$ itself. The defining series of $\delta$ and of $\tilde d$, namely the convergence exponent of
\begin{equation*}
\sum_m \bigl(1-|\langle r(\Omega_m),r(\Theta_m)\rangle|\bigr)m^{-p},
\end{equation*}
cf. Equation~(52) in \cite{LesTP}, therefore coincide term by term.

For the strong triangle inequality, Lemma~\ref{lem:merger} gives, for every third history $\boldsymbol\Xi$,
\begin{equation*}
b_m(\hist,\boldsymbol\Xi)\le 2\Bigl(b_m(\hist,\histb)+b_m(\histb,\boldsymbol\Xi)\Bigr).
\end{equation*}
Multiplication by a fixed positive constant does not change a convergence exponent, and the convergence exponent of the sum of two nonnegative sequences is the maximum of their convergence exponents. Hence
\begin{equation*}
\delta(\hist,\boldsymbol\Xi) \le \max\bigl(\delta(\hist,\histb),\delta(\histb,\boldsymbol\Xi)\bigr),
\end{equation*}
exactly as in Section~4 of \cite{LesTP}.
\end{proof}

We next show that the degeneracy of the pseudo-ultrametric is the only obstruction to completeness. Define
\begin{equation}\label{eq:delta-equivalence}
\hist\equiv_\delta\histb\;\text{ if and only if }\;\delta(\hist,\histb)=0,
\end{equation}
and let
\begin{equation*}
\widehat{\mathscr H}=X^\infty/\equiv_\delta .
\end{equation*}
Since $\delta$ is a pseudo-ultrametric, it descends to an ultrametric on $\widehat{\mathscr H}$.

\begin{theorem}[Completeness of the history space]\label{thm:complete}
The ultrametric space $(\widehat{\mathscr H},\delta)$ is complete. Neither compactness of $X$ nor continuity of the record family is used: the statement holds for an arbitrary record family on an arbitrary outcome set.
\end{theorem}

\begin{proof}
Let $(\hist^{(n)})$ be a $\delta$-Cauchy sequence of histories. Since in a pseudometric space a Cauchy sequence converges as soon as one of its subsequences does, we may pass to a subsequence, still denoted $(\hist^{(k)})$, such that
\begin{equation}\label{eq:fast}
\delta\bigl(\hist^{(k)},\hist^{(k+1)}\bigr)\le 2^{-(k+1)},\qquad k\ge1.
\end{equation}
Write
\begin{equation*}
b^{(k)}_m=b_m\bigl(\hist^{(k)},\hist^{(k+1)}\bigr)
\end{equation*}
for the consecutive defect sequences.

\noindent
\emph{Choice of blocks.}
By \eqref{eq:fast},
\begin{equation*}
\begin{split}
\cE(b^{(k)})&\le2^{-(k+1)}\\
&<2^{-k}.
\end{split}
\end{equation*}
Hence
\begin{equation*}
\sum_m m^{-2^{-k}}b^{(k)}_m<\infty.
\end{equation*}
We may therefore choose integers
\begin{equation*}
0=N_0<N_1<N_2<\cdots
\end{equation*}
such that
\begin{equation}\label{eq:tails}
\sum_{m>N_k}
m^{-2^{-k}}b^{(k)}_m
\le
8^{-k},
\qquad k\ge1.
\end{equation}
Define the pasted history $\hist^*\in X^\infty$ by
\begin{equation}\label{eq:paste}
\Omega^*_m=\Omega^{(l)}_m,
\qquad
\text{for }N_{l-1}<m\le N_l.
\end{equation}
Every coordinate of $\hist^*$ is an outcome in $X$, so $\hist^*$ is a history. In particular, no coordinatewise limit is taken: neither a closed image condition on the record map nor any topological assumption on $X$ is required.

\medskip
\noindent
\emph{Convergence along the subsequence.}
Fix $k\ge1$ and $p>2^{-k}$. We show that
\begin{equation}\label{eq:paste-series}
\sum_m b_m(\hist^*,\hist^{(k)})\,m^{-p}<\infty,
\end{equation}
which implies that
\begin{equation*}
\delta(\hist^*,\hist^{(k)})\le2^{-k}.
\end{equation*}

Split the sum into the blocks
\begin{equation*}
B_l=\{m:\,N_{l-1}<m\le N_l\}.
\end{equation*}
The blocks with $l\le k$ contain only finitely many terms, each bounded by $1$. On a block $B_l$ with $l>k$,
\begin{equation*}
b_m(\hist^*,\hist^{(k)})=b_m(\hist^{(l)},\hist^{(k)}).
\end{equation*}
Iterating the quasi-subadditivity \eqref{eq:quasi} along the chain
\begin{equation*}
\hist^{(k)},\hist^{(k+1)},\ldots,\hist^{(l)}
\end{equation*}
gives
\begin{equation}\label{eq:chain}
b_m\bigl(\hist^{(l)},\hist^{(k)}\bigr)\le\sum_{j=k}^{l-1}2^{\,j-k+1}b^{(j)}_m.
\end{equation}
Indeed, \eqref{eq:chain} follows by induction on $l-k$ from \eqref{eq:quasi}; the displayed coefficients are a convenient, slightly nonoptimal upper bound.

Since all terms are nonnegative, we may sum over the blocks and interchange the order of summation:
\begin{equation*}
\begin{split}
&\sum_{l>k}\sum_{m\in B_l}m^{-p}\,b_m\bigl(\hist^{(l)},\hist^{(k)}\bigr)\\
&\qquad\le\sum_{j\ge k}2^{\,j-k+1}\sum_{m>N_j}m^{-p}b^{(j)}_m.
\end{split}
\end{equation*}
Here, for each fixed $j$, the blocks $B_l$ with $l\ge j+1$ exhaust $\{m>N_j\}$. Since $j\ge k$,
\begin{equation*}
\begin{split}
p&>2^{-k}\\
&\ge2^{-j},
\end{split}
\end{equation*}
and therefore
\begin{equation*}
m^{-p}\le m^{-2^{-j}}.
\end{equation*}
Using \eqref{eq:tails}, we obtain \begin{equation*}
\begin{split}
\sum_{j\ge k}2^{\,j-k+1}\sum_{m>N_j}m^{-p}b^{(j)}_m&\le\sum_{j\ge k}2^{\,j-k+1}8^{-j}\\
&=2^{1-k}\sum_{j\ge k}4^{-j}\\
&<\infty.
\end{split}
\end{equation*}
Thus the exponential growth of the chaining coefficients in \eqref{eq:chain} is dominated by the geometric decay of the block tails in \eqref{eq:tails}. Together with the finite contribution from the blocks $l\le k$, this proves \eqref{eq:paste-series} for every $p>2^{-k}$, and hence
\begin{equation}\label{eq:subsequence-convergence}
\delta(\hist^*,\hist^{(k)})\le2^{-k}.
\end{equation}

\medskip
\noindent
\emph{Convergence of the full sequence.} Return now to the original Cauchy sequence, and denote by $\hist^{(n_k)}$ the $k$-th member of the chosen subsequence. Given
$\varepsilon>0$, choose $k$ so large that
\begin{equation*}
2^{-k}<\varepsilon.
\end{equation*}
Since the original sequence is Cauchy, there exists $n_0$ such that
\begin{equation*}
\delta\bigl(\hist^{(n)},\hist^{(n_k)}\bigr)<\varepsilon
\qquad
\text{for all }n\ge n_0.
\end{equation*}
By the strong triangle inequality and \eqref{eq:subsequence-convergence},
\begin{equation*}
\begin{split}
\delta(\hist^{(n)},\hist^*)&\le \max\big(\delta\bigl(\hist^{(n)},\hist^{(n_k)}\bigr),\delta\bigl(\hist^{(n_k)},\hist^*\bigr)\big)\\
&< \varepsilon
\end{split}
\end{equation*}
for $n\ge n_0$. Thus every $\delta$-Cauchy sequence converges to the $\equiv_\delta$-class of a history, and $(\widehat{\mathscr H},\delta)$ is complete.
\end{proof}

The preceding construction also survives bounded temporal dependence of the record map.

\begin{proposition}[Completeness for windowed records]
\label{prop:window-complete}
The conclusion of Theorem~\ref{thm:complete} persists for bounded window record families
\begin{equation*}
r_m=r(\Omega_{m-1},\Omega_m)
\end{equation*}
as in Remark~\ref{rem:window}, provided the block boundaries in the pasting construction are chosen to grow at least geometrically, $N_k\ge2^k$.
\end{proposition}

\begin{proof}
In the proof of Theorem~\ref{thm:complete}, the integers $N_k$ may clearly be chosen to satisfy both \eqref{eq:tails} and $N_k\ge2^k$.

For two-step windowed records, pasting corrupts the record of $\hist^*$ at exactly one site per block boundary: the site $m=N_{l-1}+1$, whose window $\bigl(\Omega^{(l-1)}_{N_{l-1}},\Omega^{(l)}_{N_{l-1}+1}\bigr)$ is a hybrid belonging to neither history. At the interior sites $m\ge N_{l-1}+2$ of $B_l$, the windowed records of $\hist^*$ and $\hist^{(l)}$ coincide. The chain estimate \eqref{eq:chain} therefore applies to the interior sites exactly as in Theorem~\ref{thm:complete}, since the quasi-subadditivity \eqref{eq:quasi} holds for any three pure record states.

The boundary sites contribute, crudely,
\begin{equation*}
\begin{split}
\sum_{l>k}\bigl(N_{l-1}+1\bigr)^{-p}&\le\sum_{l\ge k}2^{-p(l-1)}\\
&<\infty
\end{split}
\end{equation*}
for every $p>0$. Thus the boundary defects have convergence exponent zero and constitute a $\delta$-null perturbation.

Adding this boundary contribution to the interior estimate from Theorem~\ref{thm:complete} shows that, for every $p>2^{-k}$, the weighted defect series between the pasted history and the $k$-th subsequence member converges. Hence
\begin{equation*}
\delta(\hist^*,\hist^{(k)})\le2^{-k}
\end{equation*}
exactly as before, and the strong triangle inequality gives convergence of the full Cauchy sequence.

The same argument applies to any fixed window length: only finitely many sites per block boundary are corrupted, and geometric growth of the block boundaries makes their contribution summable for every $p>0$.
\end{proof}

\begin{remark}[Relation to the prefix tree metric]\label{rem:tree}
Histories also carry the usual first-disagreement ultrametric
\begin{equation*}
\delta_0(\hist,\histb)=\begin{cases}
0,
& \hist=\histb,\\[1mm]
2^{-\min\{m:\,\Omega_m\neq\Theta_m\}},
& \hist\neq\histb.
\end{cases}
\end{equation*}
The two geometries encode complementary information rather than one refining the other. The prefix metric $\delta_0$ measures \emph{when} histories first separate. The consistency rate pseudo-ultrametric $\delta$ ignores every finite prefix and measures the \emph{asymptotic
rate} at which their records separate.

Indeed, histories may agree for the first $n$ steps and then have $b_m\asymp1$ forever; as $n\to\infty$, their $\delta_0$-distance tends to zero while $\delta=1$ for every $n$. Conversely, histories can disagree from the first step with $b_m\asymp1/m$; then their prefix distance is maximal while $\delta=0$. Thus the topologies are generally incomparable. This asymptotic insensitivity to finite prefixes is precisely what makes $\delta$ a rate geometry rather than a tree location geometry.
\end{remark}

We finally relate the rate geometry to the merger quotient of Section~\ref{sec:dichotomy}. It is important that the merger relation is strictly finer than the zero distance relation of $\delta$: distinct merger classes may still lie at zero $\delta$-distance in the marginal regime.

\begin{corollary}[Descent and the coarse-graining filtration]
\label{cor:filtration}
Let $\mathscr F\subseteq X^\infty$ be a family of histories.
\begin{itemize}
\item[(i)] The pseudo-ultrametric $\delta$ descends to the merger quotient $\mathscr F/\!\sim$ of \eqref{eq:merger}. Its vanishing set beyond the diagonal is precisely the marginal stratum:
\begin{equation*}
\delta([\hist],[\histb])=0,
\qquad
[\hist]\neq[\histb],
\end{equation*}
if and only if
\begin{equation*}
\begin{split}
\sum_m b_m(\hist,\histb)&=\infty,\\
\Sigma_N(\hist,\histb)&=N^{o(1)}.
\end{split}
\end{equation*}

\item[(ii)] For each $\theta\in[0,1]$, the relation
\begin{equation*}
\delta\le\theta
\end{equation*}
is an equivalence relation on $\mathscr F/\!\sim$. Its classes, namely the closed $\delta$-balls of radius $\theta$, form a partition $\mathscr P_\theta$ of the quotient family. These partitions are nested: $\mathscr P_\theta$ refines $\mathscr P_{\theta'}$ whenever $\theta\le\theta'$. Thus every family of histories carries a canonical
one-parameter hierarchy of coarse-grainings indexed by consistency rate.
\end{itemize}
\end{corollary}

\begin{proof}
(i) Suppose $\hist\sim\hist'$. Then
\begin{equation*}
\sum_m b_m(\hist,\hist')<\infty,
\end{equation*}
and therefore
\begin{equation*}
\delta(\hist,\hist')=0.
\end{equation*}
The strong triangle inequality gives
\begin{equation*}
\begin{split}
\delta(\hist,\histb) &\le \max\big(\delta(\hist,\hist'),\,\delta(\hist',\histb)\big)\\
&=\delta(\hist',\histb).
\end{split}
\end{equation*}
Interchanging $\hist$ and $\hist'$ gives the reverse inequality, so
\begin{equation*}
\delta(\hist,\histb)=\delta(\hist',\histb).
\end{equation*}
The same argument in the second variable shows that $\delta$ is well defined on merger classes and inherits the pseudo-ultrametric properties. The description of its off-diagonal vanishing set is the partial-sum characterization of $\cE$, Equation~(13) in \cite{LesTP}.

\noindent
(ii) Reflexivity and symmetry are immediate. If
\begin{equation*}
\begin{split}
\delta(\hist,\histb)&\le\theta,\\
\delta(\histb,\boldsymbol\Xi)&\le\theta,
\end{split}
\end{equation*}
then the strong triangle inequality gives
\begin{equation*}
\begin{split}
\delta(\hist,\boldsymbol\Xi)&\le\max\big(\delta(\hist,\histb),\,\delta(\histb,\boldsymbol\Xi)\big)\\
&\le\theta.
\end{split}
\end{equation*}
Thus $\delta\le\theta$ is an equivalence relation. In a pseudo-ultrametric space its equivalence classes are precisely the closed balls of radius $\theta$, and monotonicity in $\theta$ gives the stated nested filtration.
\end{proof}

\section{The consistency rate: decay and operational meaning}
\label{sec:rate}

Throughout this section, write
\begin{equation}\label{eq:SN}
\begin{split}
S_N &=-\log\bigl|\mathfrak d_N(\hist,\histb)\bigr|\\
&= \sum_{m\le N}(-\log F_m)
\end{split}
\end{equation}
for the negative log-interference of the pair.

\begin{proposition}[Rate form of consistency]\label{prop:rate}
Suppose
\begin{equation*}
\inf_m F_m\ge c>0
\end{equation*}
along the pair. Then
\begin{equation}\label{eq:rate}
\delta(\hist,\histb)=\limsup_{N\to\infty}\,\frac{\log^+ S_N}{\log N}\,.
\end{equation}
If the pair decoheres, so that $S_N\to\infty$, this is equivalently characterized as follows: for every $\varepsilon>0$,
\begin{equation*}
S_N\le N^{\delta+\varepsilon}
\end{equation*}
for all sufficiently large $N$, while
\begin{equation*}
S_N\ge N^{\delta-\varepsilon}
\end{equation*}
along an infinite subsequence.

In particular, $\delta\in[0,1]$ always, and $\delta=1$ is the maximal possible polynomial growth exponent of the negative log-interference. Pairs satisfying
\begin{equation*}
\delta=0,
\qquad
\sum_m b_m=\infty,
\end{equation*}
form the \emph{marginal stratum} of \cite{LesTP}. They are exactly
consistent in the limit, but
\begin{equation*}
S_N=N^{o(1)},
\end{equation*}
so their normalized interference decays more slowly than every stretched
exponential $\exp(-N^\varepsilon)$, $\varepsilon>0$.

This stratum includes genuinely polynomial decay of the interference itself. For example, if
\begin{equation*}
b_m\sim\frac{a}{m},
\qquad a>0,
\end{equation*}
then
\begin{equation*}
S_N\sim a\log N,
\end{equation*}
and consequently
\begin{equation*}
|\mathfrak d_N(\hist,\histb)| = N^{-a+o(1)}.
\end{equation*}
\end{proposition}

\begin{proof}
By Proposition~\ref{prop:master}(i) - (ii),
\begin{equation*}
\Sigma_N
\le
S_N
\le
\kappa(c)\Sigma_N.
\end{equation*}
Hence $S_N$ and $\Sigma_N$ have the same polynomial growth exponent. By the partial-sum characterization of the convergence exponent, cf. Equation~(13) in \cite{LesTP},
\begin{equation*}
\begin{split}
\limsup_{N\to\infty}\,\frac{\log^+\Sigma_N}{\log N}&=\cE(b)\\
&=\delta(\hist,\histb),
\end{split}
\end{equation*}
which proves \eqref{eq:rate}. If $S_N\to\infty$, the stated upper and subsequential lower bounds are the standard characterization of this limsup.

The normalization \eqref{eq:DH} has removed the branch amplitudes, so no hypothesis on the path weights is required. Since $0\le b_m\le1$,
\begin{equation*}
\Sigma_N\le N,
\end{equation*}
and therefore $\delta\le1$.

If $\delta=0$ and $\sum_m b_m=\infty$, then $\Sigma_N\to\infty$ and hence $S_N\to\infty$, while \eqref{eq:rate} gives
\begin{equation*}
S_N=N^{o(1)}.
\end{equation*}
Finally, if $b_m\sim a/m$, then
\begin{equation*}
\begin{split}
-\log F_m&=-\log(1-b_m)\\
&\sim b_m\\
&\sim \frac{a}{m}\,,
\end{split}
\end{equation*}
so
\begin{equation*}
S_N\sim a\log N.
\end{equation*}
Exponentiating $|\mathfrak d_N|=e^{-S_N}$ gives
\begin{equation*}
|\mathfrak d_N|=N^{-a+o(1)}.
\end{equation*}
\end{proof}

\begin{remark}[Excluding perfect records]\label{rem:rate-positive}
The hypothesis $\inf_mF_m>0$ excludes the perfect record stratum of Remark~\ref{rem:perfect}. Its role is precisely to ensure that the accumulated record defect $\Sigma_N$ and the negative log-interference $S_N$ remain comparable:
\begin{equation*}
\Sigma_N\le S_N\le\kappa(c)\Sigma_N.
\end{equation*}
Thus the additive Bures functional and the suppression of normalized interference have the same polynomial scale.
\end{remark}

\begin{corollary}[Regular rate asymptotics]\label{cor:regular}
The limsup in \eqref{eq:rate} does not by itself control $|\mathfrak d_N|$ along the full sequence. If, however, the limit
\begin{equation*}
\lim_{N\to\infty}\frac{\log S_N}{\log N}
\end{equation*}
exists (for instance, if $\Sigma_N$ is regularly varying with positive index) then
\begin{equation*}
S_N=N^{\delta+o(1)}
\end{equation*}
and
\begin{equation}\label{eq:regular-decay}
|\mathfrak d_N(\hist,\histb)| = \exp\bigl(-N^{\delta+o(1)}\bigr).
\end{equation}
\end{corollary}

The consistency exponent also has a direct operational meaning. It is the growth exponent of the negative logarithm of the minimum error probability for deciding which of two cumulative records was presented.

\begin{theorem}[Optimal discrimination of histories]\label{thm:helstrom}
Let $\hist,\histb$ be a decohering pair,
\begin{equation*}
\sum_m b_m=\infty,
\end{equation*}
with $\inf_mF_m\ge c>0$, and suppose the cumulative records
\begin{equation*}
\begin{split}
R_{\hist,N} &= \bigotimes_{m\le N}r(\Omega_m),\\
R_{\histb,N} &= \bigotimes_{m\le N}r(\Theta_m)
\end{split}
\end{equation*}
are presented with equal prior probabilities. The minimum error probability over all measurements deciding between them is \cite{H76}
\begin{equation}\label{eq:helstrom}
\eP^{(N)}_{\mathrm{err}}=\frac12\Bigl(1-\sqrt{1-q_N^2}\Bigr),
\end{equation}
where
\begin{equation}\label{eq:qN-def}
\begin{split}
q_N&=\bigl|
\langle R_{\histb,N},R_{\hist,N}\rangle\bigr|\\
&=\bigl|\mathfrak d_N(\hist,\histb)\bigr|.
\end{split}
\end{equation}
Moreover,
\begin{equation}\label{eq:helstrom-rate}
\delta(\hist,\histb)=\limsup_{N\to\infty}\,\frac{\log\bigl(-\log \eP^{(N)}_{\mathrm{err}}\bigr)}{\log N}\,.
\end{equation}
Consequently, the consistency rate is the polynomial growth exponent of the negative logarithm of the minimum discrimination error for the two cumulative records.
\end{theorem}

\begin{proof}
The cumulative records are pure states. For two pure states $\rho_0,\rho_1$ presented with equal prior probabilities, Helstrom's
bound \cite{H76} gives
\begin{equation*}
\eP_{\mathrm{err}}=\frac12\Big(1-\frac12\|\rho_0-\rho_1\|_1\Big).
\end{equation*}
If the modulus of their overlap is $q$, then
\begin{equation*}
\frac12\|\rho_0-\rho_1\|_1=\sqrt{1-q^2},
\end{equation*}
which proves \eqref{eq:helstrom}. By \eqref{eq:DH}, the modulus of the cumulative record overlap is
\begin{equation*}
q_N=|\mathfrak d_N(\hist,\histb)|.
\end{equation*}

For the rate, rationalizing \eqref{eq:helstrom} gives
\begin{equation*}
\eP^{(N)}_{\mathrm{err}}=\frac{q_N^2}{2\bigl(1+\sqrt{1-q_N^2}\bigr)}\,.
\end{equation*}
Since
\begin{equation*}
1\le 1+\sqrt{1-q_N^2}\le 2,
\end{equation*}
we obtain
\begin{equation*}
\frac{q_N^2}{4}\le \eP^{(N)}_{\mathrm{err}}\le \frac{q_N^2}{2}.
\end{equation*}
Using $q_N=e^{-S_N}$, this is equivalent to 
\begin{equation}\label{eq:helstrom-SN}
-\log \eP^{(N)}_{\mathrm{err}}=2S_N+O(1).
\end{equation}

Because the pair decoheres,
\begin{equation*}
\begin{split}
q_N&\to 0,\\
S_N&\to\infty.
\end{split}
\end{equation*}
It follows from \eqref{eq:helstrom-SN} that
\begin{equation*}
\log\bigl(-\log \eP^{(N)}_{\mathrm{err}}\bigr)=\log S_N+O(1).
\end{equation*}
Dividing by $\log N$ and taking the limsup, Proposition~\ref{prop:rate} gives \eqref{eq:helstrom-rate}.
\end{proof}

\begin{remark}[Per-step optimality]\label{rem:operational}
Theorem~\ref{thm:helstrom} concerns the optimal \emph{joint} measurement on the first $N$ registers. There is a complementary per-step statement. For each pair of per-register record states, the theorem of Fuchs and Caves \cite{FC95} identifies the quantum fidelity $F_m$ with the minimum, over measurements on that register, of the classical fidelity (Bhattacharyya coefficient) of the resulting outcome distributions. Thus $b_m=1-F_m$ is the complement of the optimal classical fidelity. Pairwise, therefore, the affinity entering the classical Kakutani criterion of Remark~\ref{rem:kakutani-classical} can be realized by an optimal per-register measurement.
\end{remark}

\begin{remark}[Three readings of one number]\label{rem:three}
The same quantity $\delta(\hist,\histb)$ has now been identified in three different ways:
\begin{itemize}
\item[(i)] as the sector ultrametric $\tilde d$ of the cumulative record classes, by \eqref{eq:delta-is-td};

\item[(ii)] as the polynomial growth exponent of the negative logarithm of normalized interference, by Proposition~\ref{prop:rate};

\item[(iii)] as the polynomial growth exponent of the negative logarithm of the minimum record discrimination error, by Theorem~\ref{thm:helstrom}.
\end{itemize}
A sector distance, a decoherence rate, and a hypothesis testing rate are therefore three manifestations of the same asymptotic quantity. This is the precise sense in which the geometry of Section~\ref{sec:ultrametric} has direct operational content.
\end{remark}

\section{Classical separation profiles: a Lyapunov dictionary}\label{sec:lyapunov}

In the coherent model of Example~\ref{ex:records}(ii), the two-sided estimate \eqref{eq:coherent-defect} transfers asymptotic classical separation directly to the Bures-defect sequence and hence to the consistency rate. This yields a simple dictionary between classical separation profiles and the ultrametric strata of history space.

\begin{proposition}[Separation-to-rate dictionary]\label{prop:dictionary}
Let $d_m=d(\Omega_m,\Theta_m)$ and suppose $S$ is fixed. If $d_m^2\asymp m^{-s}$ then
\begin{equation}\label{eq:dictionary}
\begin{split}
s>1&:\ \text{merger};\\
s=1&:\ \delta=0,\; \sum b_m=\infty\ \text{(marginal)};\\
0\le s<1&:\ \delta=1-s.
\end{split}
\end{equation}
\end{proposition}

\begin{proof}
By the two-sided coherent state estimate \eqref{eq:coherent-defect}, there exist constants $c_*,c^*>0$, independent of $m$, such that
\begin{equation}\label{eq:dictionary-proof-bound}
c_*\min(1,\,Sd_m^2)\le b_m\le c^*\min(1,\,Sd_m^2).
\end{equation}
Thus the classical separation profile determines, up to multiplicative constants, the Bures defect sequence that enters both the merger criterion and the definition
of the consistency rate.

Suppose first that $s>0$. The assumption
\begin{equation*}
d_m^2\asymp m^{-s}
\end{equation*}
means that there exist constants $0<c_1\le c_2<\infty$ and an index $m_0$ such that
\begin{equation}\label{eq:dm-power-bound}
c_1m^{-s}\le d_m^2\le c_2m^{-s},\;\text{for } m\ge m_0.
\end{equation}
Since $s>0$, we have $d_m^2\to0$. Because $S$ is fixed, there is therefore $m_1\ge m_0$ such that
\begin{equation*}
Sd_m^2<1,\;\text{ for } m\ge m_1.
\end{equation*}
Consequently the truncation in \eqref{eq:dictionary-proof-bound} is eventually inactive, and for $m\ge m_1$,
\begin{equation*}
c_*S d_m^2\le b_m\le c^*S d_m^2.
\end{equation*}
Combining this with \eqref{eq:dm-power-bound}, we obtain constants $0<C_1\le C_2<\infty$ for which
\begin{equation}\label{eq:bm-power-bound}
C_1m^{-s}\le b_m\le C_2m^{-s},\;\text{ for } m\ge m_1.
\end{equation}
Thus
\begin{equation*}
b_m\asymp m^{-s}.
\end{equation*}

By the merger criterion \eqref{eq:merger}, the histories merge if and only if $\sum_m b_m<\infty$. By \eqref{eq:bm-power-bound}, this series has the same convergence behavior as
the $p$-series $\sum_m m^{-s}$. Hence, if $s>1$,
\begin{equation*}
\sum_m b_m<\infty,
\end{equation*}
and the two histories belong to the same merger class.

If $0<s\le1$, the defect series diverges, so the histories lie in distinct merger classes. Their consistency rate is $\delta(\hist,\histb)=\cE(b)$. For $s>0$, \eqref{eq:bm-power-bound} gives
\begin{equation*}
b_m m^{-p}\asymp m^{-(s+p)},
\end{equation*}
and therefore
\begin{equation*}
\sum_m b_m m^{-p}<\infty\;\text{ if and only if }s+p>1.
\end{equation*}
It follows that
\begin{equation}\label{eq:dictionary-exponent}
\begin{split}
\cE(b)&=\inf\{p\ge0:s+p>1\}\\
&=\max\{0,\;1-s\}.
\end{split}
\end{equation}

In particular, when $s=1$,
\begin{equation*}
\begin{split}
\sum_m b_m&=\infty,\\
\sum_m b_m m^{-p}&<\infty,
\end{split}
\end{equation*}
for every $p>0$, and so
\begin{equation*}
\delta(\hist,\histb)=0.
\end{equation*}
This is the marginal stratum: the histories belong to distinct merger classes even though their consistency rate vanishes. 

For $0<s<1$, \eqref{eq:dictionary-exponent} instead gives
\begin{equation*}
\delta(\hist,\histb)=1-s.
\end{equation*}

It remains to consider $s=0$. In this case $d_m^2\asymp1$, so there exist $c_1>0$ and $m_0$ such that
\begin{equation*}
d_m^2\ge c_1,\;\text{ for }\; m\ge m_0.
\end{equation*}
Equation \eqref{eq:dictionary-proof-bound} then yields the uniform lower bound: for $m\ge m_0$
\begin{equation*}
\begin{split}
b_m&\ge c_*\min(1,\,Sc_1)\\
&=:C_0.
\end{split}
\end{equation*}
Since also $b_m\le1$, we have
\begin{equation*}
C_0\le b_m\le1,
\end{equation*}
eventually. Hence $\sum_m b_m m^{-p}$ has the same convergence threshold as $\sum_m m^{-p}$: it diverges for $p\le1$ and converges for $p>1$. Therefore
\begin{equation*}
\delta(\hist,\histb)=1.
\end{equation*}
This proves \eqref{eq:dictionary}.
\end{proof}

\begin{proposition}[Mean defect criterion]\label{prop:mean}
If the pair has positive asymptotic mean defect,
\begin{equation}\label{eq:mean}
\liminf_{N\to\infty}\,\frac1N\sum_{m\le N}b_m\;>\;0,
\end{equation}
then $\delta(\hist,\histb)=1$.
\end{proposition}

\begin{proof}
By \eqref{eq:mean}, $\Sigma_N\ge cN$, for some $c>0$ and all large $N$. Consequently,
\begin{equation*}
\limsup_N\,\frac{\log^+\Sigma_N}{\log N}\ge1.
\end{equation*} 
Since $b_m\le1$ gives $\Sigma_N\le N$, the $\limsup$ is also $\le1$. By the partial sum characterization of the convergence exponent, Eq.~(13) in \cite{LesTP},  $\delta=\cE(b)=1$.
\end{proof}

A positive Lyapunov exponent does not, by itself, imply the hypothesis \eqref{eq:mean}: local exponential instability is a statement about initially nearby trajectories and does not control long-time averages. On a compact phase space, moreover, recurrences can return trajectory pairs arbitrarily close. The monotone law \(d_m\asymp(1,d_0e^{\lambda m})\), sometimes invoked in this context, is therefore a stronger pre-recurrence statement appropriate only on Ehrenfest time scales. Condition \eqref{eq:mean} replaces it by a genuine long-time criterion and cleanly separates the metric conclusion from the dynamical mechanism required to establish it.

\begin{remark}[Chaos and the maximal stratum]\label{rem:chaos}
In a sufficiently mixing chaotic model, one expects typical trajectory pairs to spend a positive fraction of time at order-one separation and hence to satisfy the mean defect condition \eqref{eq:mean}. For such pairs, $\delta=1$: consistency is reached at the maximal rate allowed by the formalism, which is the quantitative content of ``chaos produces classicality'' in this setting, cf.\ \cite{PCRR26,WS25}. Verifying \eqref{eq:mean} in a concrete model (say, the disordered kicked top) is a dynamical rather than metric task, and we do not undertake it here.
\end{remark}

Which of the separation profiles in Proposition~\ref{prop:dictionary} arise naturally from autonomous quantum dynamics is a separate dynamical question. The proposition should therefore be read as a dictionary rather than a genericity statement: once the dynamics produces a separation profile of the stated form, the coherent state estimate determines the corresponding consistency rate stratum. We return to dynamical realizability in Section~\ref{sec:discussion}.

\section{The finite size counterpart}\label{sec:pcrr}

We now make precise the relation to \cite{PCRR26}.  For $N$ spin $1/2$ degrees of freedom, the many-body Hilbert space decomposes under the commuting actions of collective spin rotations and particle permutations as
\begin{equation}\label{eq:spin permutation-decomp}
(\bC^2)^{\otimes N}=\bigoplus_S \cS_S\otimes\cP_S.
\end{equation}
Here $\cS_S\simeq\bC^{2S+1}$ carries the irreducible spin $S$ representation of $SU(2)$, while $\cP_S$ is its multiplicity space, carrying the remaining permutation degrees of freedom.  Thus $\cS_S$ describes the collective spin, whereas $\cP_S$ distinguishes the microscopically different ways in which the same collective spin representation is embedded in the $N$-spin Hilbert space.

Restricting to a fixed total-spin sector, and suppressing the subscript $S$, one therefore has
\begin{equation}\label{eq:SP-factorization}
\mathcal H_S=\cS\otimes\cP.
\end{equation}
The physical roles of the two factors in \cite{PCRR26} are quite different. The collective factor $\cS$ carries the quasiclassical degree of freedom: a point $\Omega\in\bS^2$ labels a spin coherent state $\ket{\Omega}_{\cS}$.  The multiplicity factor $\cP$, by contrast, contains microscopic degrees of freedom invisible to the collective spin coordinate.  When permutation symmetry is weakly broken, these degrees of freedom act as an internal bath and can retain records of the trajectory
followed by the collective spin.

For
\begin{equation*}
S=\frac{N-k}{2}\,,
\end{equation*}
with fixed even $k$, the multiplicity space has dimension $\dim \cP=d_\cP$ with
\begin{equation}\label{eq:P-dimension}
\begin{split}
d_\cP &= \binom{N}{k/2}-\binom{N}{k/2-1}\\
& \sim \frac{N^{k/2}}{(k/2)!}\,,
\end{split}
\end{equation}
whereas $\dim\cS=2S+1=N-k+1$.  A branch in the semiclassical regime therefore has the schematic form
\begin{equation*}
\ket{\phi_{\hist}} = \ket{\Omega_M}_{\cS}\otimes
\ket{R_{\hist}}_{\cP},
\end{equation*}
as in \eqref{eq:PCRR-branch}: the first factor specifies the final quasiclassical collective state, while the second carries the microscopic record of the history leading to it.  The weak symmetry breaking disorder couples these two sectors and produces the record fidelity estimate \eqref{eq:PCRR}.

The correspondence with our model is therefore structural rather than literal. The finite permutation bath of \cite{PCRR26} need not decompose into one fresh tensor factor per branching time.  What matters for the present theory is the accumulated record fidelity: suppression of interference is controlled by a sum of nonnegative local-in-time contributions to a squared classical separation functional. Our fresh register construction is an idealized product realization of this mechanism, in which those contributions are assigned to independent registers. Passing to infinitely many registers then allows the Kakutani - von Neumann sector alternative to replace finite dimensional near-orthogonality by exact asymptotic superselection.

This distinction raises a natural finite size question. In the infinite tensor-product model, new record capacity is supplied indefinitely as fresh registers are added. In the model of \cite{PCRR26}, by contrast, all records must coexist in the same finite dimensional multiplicity space $\mathcal P$. If $K$ different histories are to define approximately decoherent alternatives, their corresponding record vectors
\begin{equation*}
R_1,\ldots,R_K\in\mathcal P
\end{equation*}
should have small pairwise fidelities. Thus the question of how many histories a finite internal bath can distinguish becomes a geometric packing problem: how many unit vectors can be placed in $\bC^d$ while keeping every pairwise squared overlap below a prescribed tolerance $\varepsilon$?

There are two qualitatively different regimes. If the records are required to be nearly orthogonal on the scale $\varepsilon<1/d$, finite dimension imposes a strong upper bound on their number. At a fixed nonzero fidelity tolerance, however, exponentially many approximately distinguishable vectors can coexist. The following elementary proposition makes this contrast precise.

\begin{proposition}[finite dimensional record packing]\label{prop:packing}
Let $R_1,\dots,R_K$ be unit vectors in $\bC^d$, $d\ge2$, satisfying
\begin{equation*}
\max_{i\neq j}|\langle R_i,R_j\rangle|^2\le\varepsilon.
\end{equation*}
Then:
\begin{itemize}
\item[(i)] If $\varepsilon<1/d$, the Welch bound \cite{W74} gives
\begin{equation}\label{eq:welch}
K\le\frac{d(1-\varepsilon)}{1-d\varepsilon}\,.
\end{equation}
In particular, exact orthogonality implies $K\le d$.

\item[(ii)] For every fixed $\varepsilon\in(0,1)$ and every integer
\begin{equation}\label{eq:random-packing-K}
K\le (1-\varepsilon)^{-(d-1)/2},
\end{equation}
there exists a set of $K$ unit vectors with all pairwise squared fidelities at most $\varepsilon$. In particular, one may take
\begin{equation*}
K= \left\lfloor(1-\varepsilon)^{-(d-1)/2}\right\rfloor,
\end{equation*}
so the achievable packing size is exponential in $d$ at fixed threshold.
\end{itemize}
\end{proposition}

\begin{proof}
For (i), let $G$ be the Gram matrix,
\begin{equation*}
G_{ij}=\langle R_i,R_j\rangle.
\end{equation*}
Since $G\ge0$, $\operatorname{rank}G\le d$, and $\tr G=K$, the Cauchy - Schwarz inequality applied to the nonzero eigenvalues of $G$ gives
\begin{equation*}
\tr G^2 \ge \frac{(\tr G)^2}{\operatorname{rank}G} \ge \frac{K^2}{d}.
\end{equation*}
On the other hand, the assumed fidelity bound gives
\begin{equation*}
\begin{split}
\tr G^2 &=K+\sum_{i\neq j}|G_{ij}|^2\\
&\le K+K(K-1)\varepsilon.
\end{split}
\end{equation*}
Combining the two inequalities yields, for $K\ge2$,
\begin{equation*}
\frac{K-d}{d(K-1)}\le\varepsilon,
\end{equation*}
which, for $\varepsilon<1/d$, rearranges to \eqref{eq:welch}.  Thus sufficiently stringent pairwise distinguishability forces the number of records to remain of order $d$; at $\varepsilon=0$ this reduces to the familiar statement that $\bC^d$ contains at most $d$ mutually orthogonal vectors.

For (ii), draw $R_1,\ldots,R_K$ independently from Haar measure on the unit sphere of $\bC^d$. For fixed $u$ and Haar random $v$, the random variable
\begin{equation*}
t=|\langle u,v\rangle|^2
\end{equation*}
has density $(d-1)(1-t)^{d-2}$ and hence
\begin{equation*}
\eP(t>\varepsilon)=(1-\varepsilon)^{d-1}.
\end{equation*}
There are $K(K-1)/2$ unordered pairs, so the union bound gives
\begin{equation*}
\begin{split}
\eP\big(\max_{i\neq j}|\langle R_i,R_j\rangle|^2>\varepsilon\big)&\le\frac{K(K-1)}{2}(1-\varepsilon)^{d-1}\\
&<1
\end{split}
\end{equation*}
under \eqref{eq:random-packing-K}, since $K(K-1)<K^2\le(1-\varepsilon)^{-(d-1)}$. Hence an admissible configuration exists.
\end{proof}

The two parts of Proposition~\ref{prop:packing} have deliberately different logical meanings. Part~(i) is an \emph{upper bound}: when records are required to be extremely close to orthogonal, finite Hilbert-space dimension places a genuine restriction on how many can coexist. Part~(ii) is an \emph{existence result}: at any fixed nonzero fidelity tolerance, finite dimension by itself permits exponentially many pairwise distinguishable records.  It is therefore a lower bound on achievable record capacity, not a capacity ceiling.  The crossover between these statements is important: ``approximately orthogonal'' has very different dimensional consequences
depending on how the allowed fidelity scales with $d$.

This distinction can be translated into branching depth.  Suppose that each branching event has $A$ alternatives. After $M$ steps there are $K=A^M$ formal histories, and representing all of them as simultaneously distinguishable alternatives requires a corresponding family of record vectors in $\mathcal P$.  In the strict near-orthogonality regime $\varepsilon\le c/d_{\mathcal P}$ with $c<1$ fixed, the Welch bound \eqref{eq:welch} implies
\begin{equation*}
\begin{split}
A^M&\le\frac{d_{\mathcal P}}{1-c}\\
&=O(d_{\mathcal P}),
\end{split}
\end{equation*}
and hence
\begin{equation*}
M=O(\log_A d_{\mathcal P}).
\end{equation*}
For the PCRR sector with fixed $k$, $d_{\mathcal P}\sim N^{k/2}/(k/2)!$, so this stringent notion of record distinguishability permits only logarithmic branching depth in $N$.

The situation is very different if a fixed nonzero fidelity $\varepsilon\in(0,1)$ is accepted.  Part~(ii) shows that finite dimension alone does not obstruct packings satisfying
\begin{equation}\label{eq:packing-depth}
M\lesssim\frac{d_{\mathcal P}-1}{2\log A}\,\log\frac{1}{1-\varepsilon}\,.
\end{equation}
Thus a branching depth linear in the bath dimension (and therefore polynomial in $N$ for fixed $k$) is compatible with the dimensional constraint. Equation~\eqref{eq:packing-depth} should be read only as a \emph{dimensional non-obstruction}.  It does not say that the dynamics of \cite{PCRR26} actually realizes such an efficient packing, nor does it imply that record storage must fail when this scale is exceeded. The true finite size record capacity is dynamical: it depends on how records are generated, whether they recur or are overwritten, and how much pairwise fidelity is physically tolerable.

This finite size comparison also clarifies the role of the infinite tensor product in the present paper.  Fresh registers supply new record capacity indefinitely, so no fixed finite dimensional packing constraint arises. The infinite tensor product should therefore not be viewed as a literal large $N$ limit of the permutation sector $\cP$.  Rather, it is an asymptotic idealization of continued record accumulation in which the finite dimensional question ``how many nearly orthogonal records fit in the bath?'' is replaced by the Kakutani - von Neumann question ``do the accumulated record sequences belong to the same or to disjoint infinite product
sectors?''  It is this replacement that turns approximate finite size distinguishability into exact asymptotic superselection.

\section{Discussion}\label{sec:discussion}

The main result of this paper is a rate geometry for histories that accumulate quantum records. Once normalized interference factorizes into record fidelities, the infinite product alternative converts a gradual finite scale phenomenon into a sharp asymptotic one: summable Bures defect makes two record sequences weakly equivalent, whereas nonsummable defect places them in orthogonal sectors. The binary sector alternative alone, however, discards how quickly this separation occurs. The consistency rate $\delta$ retains this information and thereby refines asymptotic superselection into a geometry of histories.

The same number has three complementary interpretations. It is the gauge invariant infinite product sector pseudo-distance inherited from \cite{LesTP}; it is the polynomial growth exponent of the negative logarithm of normalized interference; and it is the corresponding exponent for the negative logarithm of the minimum Helstrom discrimination error. Thus the ultrametric structure is not merely a convenient topology imposed on histories: it is read directly from the rate at which their accumulated records become operationally distinguishable.

The coherent state example gives this abstract statement a direct dynamical interpretation. The two sided estimate \eqref{eq:coherent-defect} transfers classical separation of trajectories into quantum separation of their records. Consequently, asymptotic classical separation profiles determine the consistency rate strata. Power-law profiles $d_m^2\asymp m^{-s}$ give merger for $s>1$, the marginal regime $\delta=0$ at $s=1$, and $\delta=1-s$ for $0\le s<1$, while a positive long-time mean defect gives the maximal rate $\delta=1$. The marginal case is particularly instructive: distinct histories can belong to different superselection sectors while having zero consistency rate. Sector separation and rate separation are therefore genuinely different levels of asymptotic information.

This perspective complements existing work on the definition and emergence of branches \cite{Rie17,Wei22,TM25,RieQV25}. Those programs ask when a decomposition of a many-body wavefunction deserves to be called branching. Here we ask a different question: once record bearing branches are present, what asymptotic geometry do they carry? The answer is hierarchical: consistency rate balls give nested coarse-grainings of history space. The resulting hierarchy is insensitive to finite
modifications of a history: it measures not when two histories first separate, but the asymptotic rate at which their records become distinguishable.

The fresh-register model is deliberately idealized. Its tensor factor enumeration is physical input, not something derived from locality, and no finite record capacity is ever exhausted. The model of \cite{PCRR26} is complementary: its record bearing permutation sector emerges dynamically inside a closed many-body system, but it is finite dimensional and does not literally supply one independent register per time step. Equation~\eqref{eq:PCRR} shows that the two constructions nevertheless
share the same accumulated fidelity structure.

Section~\ref{sec:pcrr} makes the finite size distinction more precise. In a finite record space, the simultaneous storage of many distinguishable histories is a packing problem. The Welch bound limits the number of records in the near-orthogonal regime, whereas at any fixed nonzero fidelity tolerance exponentially many record vectors can coexist as the bath dimension grows. Thus finite dimension does not impose a single universal ``branching capacity'': the answer depends essentially on the required degree of record distinguishability. The infinite tensor product removes this packing constraint by construction, in the sense made precise at the end of
Section~\ref{sec:pcrr}.

An important dynamical question remains open. Proposition~\ref{prop:dictionary} is a dictionary rather than a genericity statement: it determines the consistency rate stratum once a prescribed separation profile is produced, but does not determine which profiles arise naturally from autonomous quantum dynamics. It would be particularly interesting to identify many-body systems in which the merger, marginal, intermediate rate, and maximal rate regimes emerge without explicit time dependent tuning of the system-record coupling. The isolated many-body mechanism of \cite{PCRR26} provides one natural setting in which to investigate this question. More generally, slow modes, critical or near-integrable dynamics, and dynamically generated internal record sectors may provide other realizations. Understanding which rate strata are dynamically accessible, and which are generic or exceptional, is a natural direction for future work.

A further extension is operator algebraic. Each record class determines the same type of infinite tensor product representation studied in \cite{LesTP}, so the associated factor and modular structures can in principle be transferred to history space. We have not developed that direction here. The more immediate point is operational: once histories leave cumulative quantum records, the rate at which those records become distinguishable simultaneously measures the suppression of interference and the separation of infinite product sectors. The ultrametric is therefore not an additional structure placed on the space of histories; it is the asymptotic geometry already encoded in the records themselves.

\appendix
\section{The cosine-power estimate}\label{app:cosine-power}

For completeness, we prove the elementary estimate used in Lemma~\ref{lem:cosine-power}:
\begin{equation}\label{eq:cospx-app}
(1-e^{-1})\Big(1,\;\frac{px^2}{2}\Big) \le 1-(\cos x)^p \le \Big(1,\;\frac{px^2}{2}\Big),
\end{equation}
for $|x|\le\pi/2,\,p\ge1$. Both constants are optimal.

\begin{proof}
We first prove the upper bound. Since $0\le\cos x\le1$ for $|x|\le\pi/2$, convexity of $t\mapsto t^p$, $p\ge1$, gives
\begin{equation*}
t^p\ge 1-p(1-t),\qquad 0\le t\le1.
\end{equation*}
Taking $t=\cos x$ and using $1-\cos x\le x^2/2$ yields
\begin{equation*}
1-(\cos x)^p \le p(1-\cos x)\le \frac{px^2}{2}.
\end{equation*}
Together with the trivial estimate $1-(\cos x)^p\le1$, this proves
\begin{equation*}
1-(\cos x)^p \le\Big(1,\;\frac{px^2}{2}\Big).
\end{equation*}

For the lower bound, we use
\begin{equation*}
\cos x\le e^{-x^2/2},
\qquad |x|\le\frac{\pi}{2}.
\end{equation*}
Indeed, by symmetry it suffices to take $x\ge0$. If
\begin{equation*}
f(x)=\log(\cos x)+\frac{x^2}{2},
\end{equation*}
then
\begin{equation*}
f'(x)=x-\tan x\le0,
\end{equation*}
because $\tan x\ge x$ for $x\ge0$, while $f(0)=0$. Hence $f(x)\le0$.

It follows that
\begin{equation*}
(\cos x)^p\le e^{-px^2/2},
\end{equation*}
and therefore
\begin{equation*}
1-(\cos x)^p\ge 1-e^{-px^2/2}.
\end{equation*}
Set
\begin{equation*}
v=\frac{px^2}{2}.
\end{equation*}
It remains to determine the largest constant $c$ such that
\begin{equation*}
1-e^{-v}\ge c(1,\,v),
\qquad v\ge0.
\end{equation*}
For $0<v\le1$, define
\begin{equation*}
g(v)=\frac{1-e^{-v}}{v}.
\end{equation*}
Then
\begin{equation*}
\begin{split}
g'(v)&=\frac{e^{-v}(1+v)-1}{v^2}\\
&\le0,
\end{split}
\end{equation*}
since $e^v\ge1+v$. Thus
\begin{equation*}
\frac{1-e^{-v}}{v}\ge1-e^{-1},
\qquad 0<v\le1.
\end{equation*}
For $v\ge1$,
\begin{equation*}
1-e^{-v}\ge1-e^{-1}.
\end{equation*}
Combining the two regimes gives
\begin{equation*}
1-e^{-v}\ge(1-e^{-1})(1,\,v),
\end{equation*}
and hence the lower bound in \eqref{eq:cospx-app}.

It remains to verify optimality. For the upper bound, as $x\to0$ with
$p$ fixed,
\begin{equation*}
1-(\cos x)^p=\frac{px^2}{2}+O(x^4),
\end{equation*}
so the coefficient multiplying $px^2/2$ cannot be decreased.

For the lower bound, let $x\to0$ and choose
\begin{equation*}
p=\frac{2}{x^2}.
\end{equation*}
Then $px^2/2=1$, while
\begin{equation*}
(\cos x)^{2/x^2}=\exp\Big(\frac{2\log(\cos x)}{x^2}\Big)\to e^{-1},
\end{equation*}
because
\begin{equation*}
\log(\cos x)=-\frac{x^2}{2}+O(x^4).
\end{equation*}
Consequently,
\begin{equation*}
\frac{1-(\cos x)^p}{\min(1,\,px^2/2)}\to 1-e^{-1},
\end{equation*}
so the lower constant is optimal as well.
\end{proof}

\begin{remark}
The normalization used in the spin coherent application is $a=S d^2$, rather than $v=px^2/2$.  With $p=2S$ and $x=d/2$, one has $v=a/4$.  The sharper estimate in
\eqref{eq:coherent-defect-sharp},
\begin{equation*}
(1-e^{-1/4})(1,a)\le 1-\cos^{2S}(d/2)\le(1,a),
\end{equation*}
is therefore stated and derived directly in Section~\ref{sec:model}.
\end{remark}



\begin{thebibliography}{99}

\bibitem{AW68} Araki, H., and Woods, E. J.: A classification of factors,
\textit{Publ. Res. Inst. Math. Sci.} \textbf{4}, 51 - 130 (1968).

\bibitem{Bu69} Bures, D.: An extension of Kakutani's theorem on infinite product measures to the
tensor product of semifinite $w^\ast$-algebras, \textit{Trans. Amer. Math. Soc.} \textbf{135},
199 - 212 (1969).

\bibitem{Cr97} Craig, D. A.: The geometry of consistency: decohering histories in generalized
quantum theory, arXiv:quant-ph/9704031 (1997).

\bibitem{DH92} Dowker, H. F., and Halliwell, J. J.: Quantum mechanics of history: the
decoherence functional in quantum mechanics, \textit{Phys. Rev. D} \textbf{46}, 1580 - 1609
(1992).

\bibitem{FC95} Fuchs, C. A., and Caves, C. M.: Mathematical techniques for quantum communication
theory, \textit{Open Syst. Inf. Dyn.} \textbf{3}, 345 - 356 (1995).

\bibitem{GMH90} Gell-Mann, M., and Hartle, J. B.: Quantum mechanics in the light of quantum
cosmology, in \textit{Complexity, Entropy, and the Physics of Information}, SFI Studies in the
Sciences of Complexity vol.~VIII, Addison-Wesley (1990), 425 - 458.

\bibitem{Gri84} Griffiths, R. B.: Consistent histories and the interpretation of quantum
mechanics, \textit{J. Stat. Phys.} \textbf{36}, 219 - 272 (1984).

\bibitem{Hal95} Halliwell, J. J.: A review of the decoherent histories approach to quantum
mechanics, \textit{Ann. N. Y. Acad. Sci.} \textbf{755}, 726 - 740 (1995).

\bibitem{H76} Helstrom, C. W.: \textit{Quantum Detection and Estimation Theory}, Mathematics
in Science and Engineering, vol.~123, Academic Press, New York (1976).

\bibitem{I94} Isham, C. J.: Quantum logic and the histories approach to quantum theory,
\textit{J. Math. Phys.} \textbf{35}, 2157 - 2185 (1994).

\bibitem{J94} Jozsa, R.: Fidelity for mixed quantum states, \textit{J. Mod. Opt.} \textbf{41},
2315 - 2323 (1994).

\bibitem{K48} Kakutani, S.: On equivalence of infinite product measures, \textit{Ann. Math.}
\textbf{49}, 214 - 224 (1948).

\bibitem{LesTP} Lesniewski, A.: A complete ultrametric on von Neumann's incomplete tensor products, arXiv:2607.09627 (2026). 

\bibitem{Omn92} Omn\`es, R.: Consistent interpretations of quantum mechanics,
\textit{Rev. Mod. Phys.} \textbf{64}, 339 - 382 (1992).

\bibitem{P67} Powers, R. T.: Representations of uniformly hyperfinite algebras and their
associated von Neumann rings, \textit{Ann. Math.} \textbf{86}, 138 - 171 (1967).

\bibitem{PCRR26} Pilatowsky-Cameo, S., Cotler, J., Ranard, D., and Riedel, C. J.: Emergent
classicality and wavefunction branching in an isolated quantum many-body system,
arXiv:2609.19254 (2026).

\bibitem{Rie17} Riedel, C. J.: Classical branch structure from spatial redundancy in a many-body
wave function, \textit{Phys. Rev. Lett.} \textbf{118}, 120402 (2017).

\bibitem{RieQV25} Riedel, C. J.: Wavefunction branches demand a definition!,
\textit{Quantum Views} \textbf{9}, 85 (2025).

\bibitem{S19} Schlosshauer, M.: Quantum decoherence, \textit{Phys. Rep.} \textbf{831}, 1 - 57 (2019).

\bibitem{Str24} Strasberg, P., Reinhard, T. E., and Schindler, J.: First principles numerical
demonstration of emergent decoherent histories, \textit{Phys. Rev. X} \textbf{14}, 041027 (2024).

\bibitem{TM25} Taylor, J. K., and McCulloch, I. P.: Wavefunction branching: when you can't tell
pure states from mixed states, \textit{Quantum} \textbf{9}, 1670 (2025).

\bibitem{U76} Uhlmann, A.: The `transition probability' in the state space of a $^\ast$-algebra,
\textit{Rep. Math. Phys.} \textbf{9}, 273 - 279 (1976).

\bibitem{vN39} von Neumann, J.: On infinite direct products, \textit{Compositio Math.} \textbf{6},
1 - 77 (1939).

\bibitem{WS25} Wang, J., and Strasberg, P.: Decoherence of histories: chaotic versus integrable
systems, \textit{Phys. Rev. Lett.} \textbf{134}, 220401 (2025).

\bibitem{Wei22} Weingarten, D.: Macroscopic reality from quantum complexity,
\textit{Found. Phys.} \textbf{52}, 45 (2022).

\bibitem{W74} Welch, L. R.: Lower bounds on the maximum cross correlation of signals,
\textit{IEEE Trans. Inform. Theory} \textbf{20}, 397 - 399 (1974).

\bibitem{Z03} Zurek, W. H.: Decoherence, einselection, and the quantum origins of the classical,
\textit{Rev. Mod. Phys.} \textbf{75}, 715 - 775 (2003).

\end{thebibliography}
\end{document}